\documentclass[12pt]{article}
\usepackage[square]{natbib}

\usepackage[english]{babel}
\usepackage[utf8x]{inputenc}
\usepackage{amsmath}
\usepackage{amsthm}
\usepackage{amsfonts}
\usepackage{algorithm}
\usepackage[noend]{algpseudocode}
\usepackage{graphicx}
\usepackage[colorinlistoftodos]{todonotes}
\usepackage{fullpage}
\usepackage{bbm}
\usepackage{wrapfig}
\usepackage{subcaption}
\usepackage[bottom]{footmisc}
\usepackage{tikz}

\newtheorem{theorem}{Theorem}[section]

\newtheorem{definition}[theorem]{Definition}
\newtheorem{lemma}[theorem]{Lemma}
\newtheorem{corollary}[theorem]{Corollary}

\newtheorem{thmx}{Theorem}

\newcommand{\R}{\mathbb{R}}

\newcommand{\eps}{\varepsilon}
\newcommand{\tw}{\tilde{w}}
\newcommand{\tR}{\tilde{R}}

\renewcommand{\r}{r}
\renewcommand{\P}{\mathbb{P}}
\DeclareMathOperator{\E}{\mathbb{E}}

\newcommand{\D}{\mathcal{D}}

\def \bX {\mathbf{X}}
\def \bY {\mathbf{Y}}
\def \bZ {\mathbf{Z}}

\newcommand{\poly}{\mathrm{poly}}
\newcommand{\MAX}{\mathsf{MAX}}
\newcommand{\OPT}{\mathsf{OPT}}
\newcommand{\Obj}{\mathsf{Obj}}
\newcommand{\ALG}{\mathsf{ALG}}
\newcommand{\free}{\mathsf{free}}
\newcommand{\bbig}{\mathsf{big}}

\newcommand{\rand}{\mathsf{rand}}

\title{Variable Decomposition for\\ Prophet Inequalities and Optimal Ordering}
\author{Allen Liu\thanks{Massachusetts Institute of Technology, email: cliu568@gmail.com} \and Renato Paes Leme\thanks{Google Research, email: renatoppl@google.com}  \and Martin P\'{a}l\thanks{Google Research, email: mpal@google.com} \and Jon Schneider\thanks{Google Research, email: jschnei@google.com} \and Balasubramanian Sivan\thanks{Google Research, email: balusivan@google.com}}
\date{}

\begin{document}

\maketitle
\thispagestyle{empty}

\begin{abstract}
We introduce a new decomposition technique for random variables that maps a generic instance of the prophet inequalities problem to a new instance where all but a constant number of variables have a tractable structure that we refer to as $(\eps, \delta)$-smallness. Using this technique, we make progress on several outstanding problems in the area:

\begin{itemize}
    \item We show that, even in the case of non-identical distributions, it is possible to achieve (arbitrarily close to) the optimal approximation ratio of $\beta \approx 0.745$ as long as we are allowed to remove a small constant number of distributions.  %
    \item We show that for \textit{frequent} instances of prophet inequalities (where each distribution reoccurs some number of times), it is possible to achieve the optimal approximation ratio of $\beta$ (improving over the previous best-known bound of $0.738$).
    \item We give a new, simpler proof of Kertz's optimal approximation guarantee of $\beta \approx 0.745$ for prophet inequalities with i.i.d. distributions. The proof is primal-dual and simultaneously produces upper and lower bounds.
    \item Using this decomposition in combination with a novel convex programming formulation, we construct the first Efficient PTAS for the Optimal Ordering problem.%
\end{itemize}

\end{abstract}

\newpage
\setcounter{page}{1}

\section{Introduction}

The problem of prophet inequalities is a classic problem in optimal stopping theory consisting of designing stopping times for sequences of independent random variables that always achieve (in
expectation) a certain fraction of the maximum value in the sequence. Formally, it
asks for the maximum constant $\alpha$ such that for any independent
random variables  $X_1, \hdots, X_n$ with known distributions, it is possible to design a stopping time $\tau$ such that $\E[X_\tau] \geq \alpha \cdot \E[\max_i X_i]$. 

Another way to phrase this problem is to imagine $n$ boxes: each is labelled
with a distribution $F_i$ and inside the box there is a random variable $X_i$ distributed independently according to $F_i$. A gambler is then presented with the boxes in some sequence. After opening each box, the algorithm inspects the value of $X_i$ and either takes it and ends the game or discards it (discarded items are forever lost) and continues opening further boxes. The goal is to design algorithms for the gambler that are competitive against a \emph{prophet} who sees the content of all the boxes without having to open them, and therefore just picks the box containing the maximum variable. We will be concerned with two main variations: (i) \emph{free order}, where the algorithm is allowed to select the order in which boxes are inspected and; (ii) \emph{random order} where boxes are presented in uniformly random order. This random order model is sometimes referred to as the ``prophet secretary'' model (see e.g. \cite{esfandiari2017prophet,EHKS18,CSZ19}).

An important special case is the case where the distributions in all boxes are identical. In this case, the free order and random order versions are identical. An upper bound for this case was given by  \cite{kertz1986}
who showed that no algorithm can obtain an approximation ratio better than the solution $\beta
\approx 0.745$ to the equation:
\begin{equation}\label{eq:kertz}
\int_0^1 [(\beta^{-1}-1) - y (\log y - 1)]^{-1} dy = 1
\end{equation}
\cite{CFHOV17} recently demonstrated a policy matching this upper bound, closing the i.i.d. question. This established $\beta \approx 0.745$ as the golden standard for prophet inequalities and a major open question is whether the same bound can be achieved in more complex settings, like non-identical distributions.

\paragraph{Main technique: small variables decomposition.} In this paper, we introduce a new decomposition technique that maps an instance of prophet inequalities to another instance where all but a constant number of random variables have a tractable structure, encoded by the notion of $(\eps,\delta)$-smallness. A random variable $X$ is $(\eps,\delta)$-small if $\Pr[X \geq \delta] \leq \eps$. These are variables with most of their probability mass concentrated near zero.

The decomposition works as follows: given any instance $\bX = \{X_1, X_2, \dots, X_n\}$ of random variables and a threshold $t > 0$ we consider a new set of random variables formed by $Z_i := \max(X_i, t) - t$. As we increase $t$, more and more of these random variables become small. Our algorithms then follow a similar template: if we ever see an item larger than $t$, we accept it; otherwise, we run an algorithm on these residues exploiting the fact that they are small.

This decomposition enables various new algorithmic results like establishing the tight approximation ratio for various settings (e.g. frequent instances and imperfect prophets) as well as an efficient PTAS for the Optimal Ordering Problem. Another byproduct is a simpler proof of $\beta \approx 0.745$ for i.i.d. settings where we simulateously establish the upper and lower bounds. Below, we describe the main building block together with our results. 

\paragraph{Building block: Prophet inequalities for small variables.} We start by establishing prophet inequalities when all variables are $(\eps,\delta)$-small:
\begin{thmx}[Restatement of Corollary \ref{cor:eps-delta-small}]\label{thm:iid_intro}
Given a set of $(\eps,\delta)$-small variables $\bX = \{X_1,\dots, X_n\}$, there exists a stopping rule $\tau$ such that
$\E[X_{\tau}] \geq (\beta-O(\eps)) \cdot \E[\max_i X_i] - \delta.$, where $\beta \approx 0.745$.
\end{thmx}

To prove Theorem \ref{thm:iid_intro}, we design what we call \textit{time based policies} for small prophets instances. In such a policy, we interpret each random variable arriving at a uniform random time in $[0, 1]$. Then, in order of arrivals, we choose whether to accept each variable based on whether it is above a time-based threshold (e.g. if variable $X$ arrives at time $t$, we select it if $X \geq r(t)$ for some function $r$). For small prophets, we can show that the performance of the optimal such time-based policy is exactly the Kertz bound.\\

Before discussing the new results we obtaim using Theorem~\ref{thm:iid_intro}, we note here that we obtain as a corollary of Theorem~\ref{thm:iid_intro} a simple $\beta \approx 0.745$-approximation for i.i.d. variables (not necessarily small). The proof obtained via this construction is primal-dual and simultaneously obtains the upper and lower bound on the approximation factor.

\paragraph{Imperfect prophets}
A major open problem in prophet inequalities has been to determine whether a $\beta$ approximation is possible for non-identical distributions. It remained open for a while for both free order and random order settings. For the random order case, this question was recently resolved by~\cite{CSZ19}, who show an upper bound on the approximation ratio of $\sqrt{3} - 1 \approx 0.732$, strictly less than the Kertz bound of $0.745$. However, we show that if we eliminate a constant number of variables, the gap between i.i.d. and non-i.i.d. disappears and we can obtain a $0.745$ approximation.

We demonstrate a policy that achieves the approximation ratio of $\beta - \varepsilon$ when compared to the expected max of all but a constant number of boxes (a constant depending only on $\varepsilon$ but not on $n$) . In fact, we prove something slightly stronger: that by removing $\poly(\eps^{-1})$ random variables from an instance, we can find a stopping time for the resulting instance that is $(\beta - \eps)$-competitive with the prophet benchmark.

\begin{thmx}[Restatement of Theorem \ref{thm:nearby}]
\label{thm:nearby_main}
Consider a collection of $n$ random variables $\bX = \{X_1, X_2, \dots, X_n\}$. Then, for any $\eps > 0$, there exists a subset $\bX'$ of $\bX$ containing $n - \Tilde{\Theta}(\eps^{-3})$ of these variables and a stopping rule $\tau$ for $\bX'$ such that if the items arrive in a uniform random order, $\E[X'_{\tau}] \geq (\beta - \eps)\E[\max_i X'_i]$. The subset $\bX'$ and a stopping rule $\tau$ achieving this guarantee can be constructed in polynomial time. 
\end{thmx}

One way to interpret this result is that all instances of this problem are ``close'' to instances where it is possible to achieve the Kertz bound: the only obstructions to achieving the Kertz bound for non-identical distributions are some small (constant-sized) sets of distributions.

Theorem \ref{thm:nearby_main} has a number of interesting applications. For example, Theorem \ref{thm:nearby_main} immediately implies an asymptotically tight prophet inequality for the $\Tilde{\Theta}(\eps^{-3})$th order statistic (Corollary \ref{cor:kthlargest}). This is what we call the \emph{imperfect prophet} case, where the benchmark picks the $k$-largest variable instead of the largest. This is equivalent to the $k$-optimal price benchmark in \cite{goldberg2006competitive}.

\paragraph{Frequent Instances}
Theorem \ref{thm:nearby_main} also implies improved results for \textit{frequent prophets}. Frequent prophets, introduced in \cite{AEEHKL17}, are a subclass of prophet inequality instances where each distribution must be repeated at least some number of times; an instance is $m$-frequent if each distribution appears at least $m$ times. In \cite{AEEHKL17}, the authors show how to design a $0.738$-competitive policy for $\Theta(\log n)$-frequent instances. We improve on this by showing how to obtain a $(\beta - \eps)$-competitive policy for $O_{\eps}(1)$-frequent instances. 

\begin{thmx}[Restatement of Theorem \ref{thm:frequent}]
\label{thm:frequent_intro}
Consider an $m$-frequent collection of independent random variables. If $m = \Tilde{\Omega}(\eps^{-2})$ for some $\eps >0$, it is possible to construct a stopping time $\tau$ such that $\E[X_\tau] \geq (\beta - \eps) \E[\max_i X_i]$, where $\beta$ is the Kertz bound. A policy achieving this guarantee can be constructed in polynomial time. 
\end{thmx}

Both the bounds in Theorems~\ref{thm:nearby_main} and~\ref{thm:frequent_intro} are tight. I.e., it is impossible to obtain better-than-$\beta$ approximation removing any constant (independent of $n$) number of variables in Theorem~\ref{thm:nearby_main} as the i.i.d. instance is a special case where it is known that we can't obtain better-than-$\beta$ approximation. Likewise the frequent prophets in Theorem~\ref{thm:frequent_intro} also have the i.i.d. instance as a special case.

\paragraph{An Efficient PTAS (EPTAS) for the Optimal Ordering problem} The Optimal Ordering Problem has the same setting as free-order prophet inequalities: a collection of random variables with known distributions that the algorithm can inspect in any order of its choice. Instead of competing with the prophet benchmark, however, we compete against the optimal online policy. The main complexity is finding the optimal order in which to inspect the variables. For each fixed order it is simple to derive the optimal policy by backwards induction. The Optimal Ordering problem was shown to be NP-hard by \cite{agrawal2020optimal} who also give an FPTAS if all distributions have support size at most $3$. If all the distributions have constant support, a PTAS is given in \cite{fu2018ptas}. \cite{chakraborty2010approximation} give a PTAS for the Posted Prices problem  without any assumption on the support of the distribution. This PTAS can be translated to a PTAS for the Optimal Ordering problem via the reduction in \cite{correa2019pricing}. The PTAS in~\cite{chakraborty2010approximation} has running time $O(n^{\poly(1/\eps)})$ and hence it is a PTAS but not an Efficient PTAS\footnote{An algorithm is a PTAS if its running time $f(n,\eps)$ is a polynomial in $n$ for any fixed $\eps>0$. In particular the exponent of $n$ may increase with $\eps$. An EPTAS (as defined in \cite{cesati1997efficiency}) is an algorithm with running time of the form $f(\eps) \poly(n)$, i.e. where $\eps$ can affect the constants, perhaps exponentially, but not the exponent of $n$. An EPTAS sits in between a FPTAS and a PTAS.} (EPTAS).

Using our decomposition technique, we give the first EPTAS with running time $O(\exp(\eps^{-O(1)}) \cdot \poly(n))$ for the Optimal Ordering problem without any assumption on the distributions. We only assume that we can query $\P(X \leq x)$ for every variable $X$ and any constant $x$. As a first step, we give an $O(\poly(n, 1/\eps))$-algorithm that achieves an $1-\eps$ approximation for the Optimal Ordering variables if all the variables are $\eps$-small. The algorithm involves a novel formulation of the Optimal Ordering Problem as a concave program coupled with a randomized rounding technique. We then use our variable decomposition technique to reduce the problem to an instance with at most a constant number of variables that are not small. Using this, we can guess order of such variables and then insert the remaining variables using a combination of the convex program and a dynamic program.

\subsection{Related Work}

The literature on prophet inequalities \citep{KS77} is vast; we provide here a high-level overview of the prophet inequalities landscape, primarily focusing on the case where a single random variable out of $n$ needs to get picked. 
Three important variants of the prophet inequality problem are adversarial order prophets, random order prophets and free order prophets. As the names imply, the random variables arrive respectively in an adversarial order, uniformly random order, and in an order of the algorithm's choice in these three settings. 
   
\paragraph{Adversarial order.} 
When the random variables are independent but not identical and they arrive in an adversarially chosen order,~\citet{KS78} showed that the gambler can obtain at least $\frac{1}{2}$ of the prophet benchmark. Later~\cite{samuel-cahn1984} showed that the same $\frac{1}{2}$-approximation can be obtained by a simple threshold policy, that posts a single threshold and accepts the first random variable to exceed the threshold\footnote{This constant $\frac{1}{2}$ cannot be improved, even if the algorithm was allowed to use adaptive strategies and even if we were in the $m$-frequent prophet setting for arbitrarily large $m$.}.
In the special case where the random variables are i.i.d. (and hence the adversarial, random and free orders coincide), 
~\cite{HK82} show that the gambler can obtain at least 
$1-\frac{1}{e}$ of the prophet benchmark and also show examples that prove that 
one cannot obtain a factor beyond $\frac{1}{1.342}\sim 0.745$.~\citet{kertz1986} 
later conjecture  that $\frac{1}{1.342}\sim 0.745$ is the best possible 
approximation. The first formal proof that one can go beyond $1-\frac{1}{e}$ 
was given by~\cite{AEEHKL17} and~\cite{CFHOV17}. In the former,~\citeauthor{AEEHKL17} give a $0.738$ approximation when $n$ is larger than a large constant $n_0$. Simultaneously and independently,~\cite{CFHOV17} show a 
$0.745$ approximation for this problem, thereby completely closing the gap between upper and lower bounds for the i.i.d. case. 

\paragraph{Random order.} The random order prophet inequality problem, or the prophet secretary problem, was first studied by~\cite{esfandiari2017prophet}, where they show a $1-\frac{1}{e}$ approximation for large $n$. They also show that with a single threshold it is impossible to get better than a 
$\frac{1}{2}$-approximation.~\cite{CFHOV17} were the first to show that one can obtain a $1-\frac{1}{e}$ approximation for any $n$, and they do this via non-adaptive thresholds.~\cite{azar2017prophet} were the first to beat the $1-1/e$ barrier for the random order prophet problem and show an approximation factor of $1-\frac{1}{e} + 0.0025$. The notion of time based policies that we crucially use in our paper (see Section~\ref{sec:time-based}) was introduced in~\cite{EHKS18} to obtain a $1-1/e$ approximation for random order prophets with matroid feasibility constraints (i.e., the set of random variables the gambler can feasibly pick has to be an independent set of an uunderlying matroid). These time based policies were used later by~\cite{singla2018combinatorial} for obtaining an alternative proof of $0.745$ approximation for the i.i.d. case (obtained originally by~\cite{CFHOV17}), although this proof works only for very large $n$. The time based policies were subsequently used by~\cite{CSZ19} to show a $0.669$ approximation factor which remains the best known factor to date for the random order prophet problem. In the same paper~\citeauthor{CSZ19} show that it is impossible to get larger than $\sqrt{3}-1$ approximation factor. When each distribution occurs at least $\Theta(\log n)$ times,~\cite{AEEHKL17} show that one can obtain a $0.738$ approximation for random order prophets.

\paragraph{Free order.} %
While the upper bound of $\sqrt{3}-1$ from~\cite{CSZ19} is not known to hold in the free order case, the $0.669$ bound established by~\citeauthor{CSZ19} in the same paper for the random order case continues to be the best known approximation factor for the free order case as well, for general $n$. There are a few special cases where better approximation factors are known. For small $n$,~\cite{BGPPS18} design better approximation factors through factor revealing LPs. When each distribution occurs at least a constant $m_0$ times,~\cite{AEEHKL17} show that one can obtain a $0.738$ approximation for free order prophets.

\section{Prophet inequalities setting}

In the prophet inequalities setting we have $n$ random variables $X_i$ taking non-negative real values with known
distributions $F_i$ arriving in random order. Upon the arrival, an algorithm
learn its realization and decides to either stop and obtain that value as reward or reject that variable and 
continue.

It is useful to think of the random variables as $n$ boxes each labeled with the
distribution $F_i$. Inside each box is a sample $X_i \sim F_i$. The distributions
are known in advance, but the boxes are given to the algorithm in random order.
The identity of the $i$-th box is only revealed when that box arrives. Upon arrival,
the algorithm can inspect the content of that box.

Given a collection of random variables $\bX = \{X_1, \hdots, X_n\}$ we will denote
by $\OPT(\bX)$ the reward of the optimal policy when the variables arrive in a uniform random order. This will be compared with a
\emph{prophet} who can see the value inside all the boxes in advance. The
reward of the prophet is given by $\MAX(\bX) = \E[\max_{i=1..n} X_i]$. The prophet inequality problem asks how good is the optimal online policy when compared with the prophet. In other words, what is the largest factor $\alpha$ for which the following inequality is true: 
$$\OPT(\bX) \geq \alpha \cdot \MAX(\bX)$$
We will often omit $\bX$ and just refer to $\OPT$ and $\MAX$ when clear from context. For bounding $\OPT$ we will often define a certain feasible policy $\ALG$ and then bound it with respect to $\MAX$. This will immediatly imply a prophet inequality since $\OPT \geq \ALG$ for any feasible online policy.

\subsection{Useful definitions}

We present certain definitions that will be useful throughout the paper:

\begin{definition}[$(\eps,\delta)$-Small Variables]\label{def:small}
We say that a random variable $X_i \geq 0$ with c.d.f. $F_i$ is $(\eps,\delta)$-small if
$$1- F_i(r) \leq \eps, \quad \forall r > \delta$$
or in other words, $1-\eps$ of its mass is in $[0,\delta]$. We say that a variable is
$\eps$-small if it is $(\eps, 0)$-small.
\end{definition}

\begin{definition}[Imperfect Prophet]\label{def:imperfect}
We say that for a set of random variables $\bX$, the Imperfect Prophet benchmark $\MAX_k(\bX)$ corresponds to the expectation of the $k$-th largest value of $X_i$.
\end{definition}

\begin{definition}[Free-order policy]\label{def:free_order}
The free-order optimum $\OPT_{\free}(\bX)$ is the reward of the optimal online policy that can choose the order in which the boxes are inspected. 
\end{definition}

\begin{definition}[$m$-frequent variables]\label{def:m_frequent}
We say that a set of random variables $\bX = \{X_1, \hdots, X_n\}$ is $m$-frequent if for any variable $X_i$ there are at least other $m-1$ variables with the same distribution.
\end{definition}

\subsection{Kertz upper bound}

Kertz shows that even if the random variables are i.i.d. (in which case the order in which boxes are inspected is irrelevant) the maximum possible factor in prophet inequalities is $\beta \approx 0.745$:

\begin{theorem}[\cite{kertz1986}]\label{thm:kertz_upper_bound}
Let $\beta$ be the solution to Kertz's equation \eqref{eq:kertz}. Then for every $\varepsilon$ there is a set of i.i.d. random variables $\bX$ such that $\OPT(\bX) \leq (\beta+\varepsilon) \cdot \MAX(\bX)$.
\end{theorem}

\section{Prophet inequalities for small variables}\label{sect:small}

Our first step is to prove a prophet inequality for small variables (Definition \ref{def:small}). We first argue that the Kertz upper bound (Theorem \ref{thm:kertz_upper_bound}) is still valid when restricted to $\varepsilon$-small distributions. Then we give a policy for small prophets achieving $\beta-O(\eps)$.

\begin{lemma}\label{lemma:kertz_small}
For any $\eps,\delta>0$ there is a set $\bX = \{X_1, \hdots, X_n\}$ of $\eps$-small variables such that $\OPT(\bX) \leq (\beta + \delta) \cdot \MAX(\bX)$ where $\beta$ is the Kertz bound.
\end{lemma}

The proof is based on the following observation. This and other omitted proofs can be found in the appendix.

\begin{lemma}\label{lemma:iid_splitting}
Let $\bX$ be a set of $n$ i.i.d. variables with c.d.f. $F$ and $\bY$ be a set of $nk$ variables with c.d.f. $F^{1/k}$. Then $\MAX(\bX) = \MAX(\bY)$  and $\OPT(\bX) \geq \OPT(\bY)$.
\end{lemma}

This lemma allows us to convert an upper bound of  $\OPT(\bX) \leq (\beta + \delta) \cdot \MAX(\bX)$ for c.d.f. $F$ to an upper bound of type $\OPT(\bY) \leq (\beta + \delta) \cdot \MAX(\bY)$ for c.d.f. $F^{1/k}$. If we add a tiny probability mass at zero and take $k$ to be large enough, the distribution $F^{1/k}$ becomes $\eps$-small. A formal proof is given in the appendix.\\

Our main result in this section is an algorithm achieving the optimal bound. 

\begin{theorem}\label{thm:small_prophets}
Given a (non-i.i.d.) set of $\eps$-small variables $\bX$, then  $$\OPT(\bX) \geq (\beta-O(\eps)) \cdot \MAX(\bX),$$ where $\beta$ is the Kertz bound.
\end{theorem}

The proof of this theorem will be given in the following subsections. Before presenting the proof we would like to point out two consequences.
Firstly, we can obtain the optimal algorithm for i.i.d. (but not necessarily small) variables as a corollary, providing an alternate proof of the result by \cite{CFHOV17}:

\begin{corollary}
If $\bX$ is a set of i.i.d. variables (not necessarily small), then $$\OPT(\bX) \geq \beta \cdot \MAX(\bX)$$
\end{corollary}

\begin{proof}
Consider an i.i.d. prophet instance with $n$ random variables $X_1, X_2, \dots, X_n$ each with c.d.f. $F$. As in the proof of Lemma \ref{lemma:kertz_small}, perturb $F$ slightly so that $F(0) > 0$ and choose $k$ such that $F^{1/k}(0) \geq 1-\eps$. Consider the $nk$ variables $Y_1, Y_2, \dots, Y_{kn}$. They are $\eps$-small, so by Theorem \ref{thm:small_prophets}, $\OPT(\bY) \geq \beta \cdot (1-O(\varepsilon)) \cdot \MAX(\bY)$. Lemma \ref{lemma:iid_splitting} then implies that $\OPT(\bX) \geq  (\beta-O(\varepsilon)) \cdot \MAX(\bX)$. This is true for any $\eps$ and hence $\OPT(\bX) \geq \beta \cdot \MAX(\bX)$
\end{proof}

Finally, we show that Theorem \ref{thm:small_prophets} extends to $(\eps, \delta)$-small prophets at the cost of an additive $\delta$. This fact will be useful in Section \ref{sect:imperfect}.

\begin{corollary}\label{cor:eps-delta-small}
Given a (non-i.i.d.) set of $(\eps,\delta)$-small variables $\bX$, then  $$\OPT(\bX) \geq (\beta-O(\eps)) \cdot \MAX(\bX) - \delta$$
\end{corollary}

\begin{proof}
  Apply Theorem \ref{thm:small_prophets} to the variables $\tilde{X}_i = X_i \cdot {\bf 1} \{X_i > \delta\}$.
\end{proof}

\subsection{Time Based Policies}
\label{sec:time-based}
The policy we will construct in the proof of Theorem \ref{thm:small_prophets} has the form of a time based policy, which we describe below.
Let $\pi$ be the random permutation of the boxes, i.e., the $i$-th
box inspected by the algorithm is $F_{\pi(i)}$. It is useful to think of the random arrivals of
boxes in terms of timestamps, i.e., that each box arrives uniformly at random in
a time $t \in [0,1]$. We will assign timestamps to boxes that are i.i.d.
uniform and consistent with $\pi$. This can be done in the following manner:

\begin{enumerate}
  \item Sample timestamps $t'_1, \hdots, t'_n$ i.i.d. from the uniform distribution over
    $[0,1]$. Sort the numbers and let $t'_{(i)}$ be the $i$-th smallest sampled
    timestamp.
  \item Given a permutation $\pi$, assign timestamp $t'_{(i)}$ to
    box $\pi(i)$ by setting $t_{\pi(i)} = t'_{(i)}$.
\end{enumerate}

\begin{lemma} The variables $t_1, \hdots, t_n$ are i.i.d. uniform.
\end{lemma}

\begin{proof}
  The variables are obtained by sampling i.i.d. uniform random variable and then
  applying a random permutation, hence the result must be also i.i.d. uniform.
\end{proof}

The timestamps encode the permutation in which the boxes arrive, so from this
point on we will reason solely in terms of timestamps. Note that in the original problem there is no notion of time, only arrival order, but the construction above (sample timestamps, sort them and assign timestamps in the order of arrival) allows us to think in terms of time arrivals.

\paragraph{Time based threshold} The policy we will consider is parametrized by a decreasing threshold function $r:[0,1] \rightarrow \R_+$. If a variable arrives at time $t$, we will pick that variable with probability $(1-\eps)^2$ if $X_i \geq r(t_i)$ and not pick it otherwise.

\subsection{Notation and Useful Inequalities}\label{sec:small_notation}

We start by establishing some notation and some useful inequalities. We will use the notation $\bar F_i(r)$
to denote $1-F_i(r)$. The c.d.f. of $\max_i X_i$ is $F(r) := \prod_i F_i(r)$. For any threshold $r$ note that:
\begin{equation}\label{eq:Ri}
\tR_i(r) := \E[X_i \cdot {\bf 1}\{X_i \geq r\}] = r \bar F_i(r) + \int_r^\infty \bar F_i(s) ds
\end{equation}
In various places of the proof, it will be useful to analyze the sum of the above quantity over all variables and the following notation will come in handy:
\begin{equation}\label{eq:tw}
  \tw(r) = \sum_i \bar F_i(r)  
\end{equation}
\begin{equation}
  \tR(r) := \sum_i \E[X_i \cdot {\bf 1}\{X_i \geq r\}] = r \tw(r) + \int_r^\infty \tw(s) ds 
\end{equation}
Whenever $\bar F_i(r)$ is small it will be convenient to approximate $\tw(r)$ by:
\begin{equation}\label{eq:w}
  w(r) :=  -\sum_i \log (1-\bar F_i(r))  = -\log F(r) 
\end{equation}
For $\varepsilon \leq 0.6$ it holds that $x \leq -\log(1-x) \leq
(1+\eps) x$ for $x \in [0,\eps]$. Since all variables are $\varepsilon$-small, we have:
\begin{equation}\label{eq:wapprox}
\tw(r) \leq w(r) \leq (1+\eps) \cdot \tw(r)
\end{equation}
It will be equally convenient to define an approximation of $\tR$ using $w$ as follows:
\begin{equation}\label{eq:defR}
R(r) = r w(r) + \int_r^\infty w(s) ds = -\int_r^\infty s w'(s) ds
\end{equation}
By equation \ref{eq:wapprox} we have:
\begin{equation}\label{eq:Rapprox}
\tR(r) \leq R(r) \leq (1+\eps) \cdot \tR(r) 
\end{equation}

We will use $\ALG$ to denote the performance of the policy that picks a variable with probability $(1-\eps)^2$ whenever $X_i \geq r(t_i)$. In the next subsection we will establish bounds on $\ALG$ and $\MAX$. We will use $\ALG$ as a lower bound for $\OPT$.

\paragraph{Warning on differentiability} Throughout this section we will assume that the c.d.f. $F$ and $F^{-1}$ are both differentiable in their respective domains. This can be done without loss of generality by the following argument: replace each variable $X_i$ by $X_i \cdot M_i$ where $M_i$ is a random variable with $\mathcal{C}^\infty$ c.d.f. supported $[0,1]$ which is in $[1-\delta,1]$ with probability $1-\delta$ for some very small $\delta$. By the standard ``mollifier'' argument, $X_i \cdot M_i$ will have $\mathcal{C}^\infty$ c.d.f. and inverse c.d.f. The results can be proven for such variables and then ported to the original variables by taking $\delta$ to zero.

\subsection{Bounding the Prophet and the Algorithm}

\begin{lemma}[Prophet bound]\label{lemma:prophet_bound}
  The prophet benchmark can be written as:
  $$\MAX = \int_{0}^\infty R(r) e^{-w(r)} w'(r) dr$$
\end{lemma}

\begin{proof} Since $F$ is the c.d.f. of $\max_i X_i$ we have:
  \begin{equation}\label{eq:pb}
  \MAX = \int_{0}^\infty (1-F(r)) dr = \int_{0}^\infty \left(1-e^{-w(r)}\right) dr =
  - \int_{0}^\infty r w'(r) e^{-w(r)}  dr
  \end{equation}
  where the last equality follows by integration by parts. Integrating by parts
  again we obtain the result in the statement.
\end{proof}

We now lower bound the performance of the time-based threshold policy:

\begin{lemma}[Policy bound]\label{lemma:policy_bound}
Given a non-increasing threshold function $r:[0,1] \rightarrow \R_+$ the performance of the policy that picks a variable with probability $(1-\eps)^2$ whenever $X_i \geq r(t_i)$ is lower bounded as follows:
  $$\ALG \geq (1-\eps)^3 \int_0^1 R(r(t))  \exp\left( - \int_0^t w(r(s)) ds \right) dt$$
\end{lemma}

\begin{proof}
We prove the bound in three steps:\\

\emph{Step 1:} Let's look at a single variable $i$ and compute the probability
it is not picked before time $t$. This is at least
$$1-(1-\eps)^2 \cdot \int_0^t \bar F_i(r(s)) ds \geq \exp \left(
-(1-\eps) \cdot \int_0^t \bar F_i(r(s)) ds  \right)$$
using the fact that for $x \in [0,\eps]$ we have $1-x(1-\eps) \geq
e^{-x}$.\\

\emph{Step 2:}  The probability that no variable is picked before time $t$ is at least
$$\prod_i  \exp \left( -(1-\eps) \cdot \int_0^t \bar F_i(r(s)) ds  \right) = 
\exp \left( -(1-\eps) \cdot \int_0^t \tw(r(s)) ds \right) \geq  \exp \left(-
\int_0^t w(r(s)) ds\right).$$\\

\emph{Step 3:} Now if variable $i$ arrives in interval $[t,t+dt]$ the reward is $(1-\eps)^2 \tR_i(r(t))$ (see equation \ref{eq:Ri}) times the probability that none of the other variables have already been previously picked. This probability has been lower bounded in the previous step. Integrating, we obtain that the expected reward from variable $i$ is at least:
$$  (1-\eps)^2 \int_0^1 \tR_i(r(t))  \exp\left( - \int_0^t w(r(s)) ds \right) dt $$
Summing over all $i$ and using that $\sum_i \tR_i(r(t)) = \tR(r(t)) \geq (1-\varepsilon)  R(r(t))$ (see equation \ref{eq:Rapprox}) we get the bound in the lemma.
\end{proof}

\subsection{Relating the bounds to Kertz's equation}

Our next step is to reduce the problem to solving a differential equation:

\begin{lemma}\label{lemma:yFr}
  If $y:[0,1]\rightarrow [0,1]$ is a function satisfying 
\begin{equation}\label{eq:prob_a}
 \exp\left(  \int_0^t \log y(s) ds \right) = -\beta \cdot y'(t), \quad y(0) = 1
  \quad \text{and} \quad y(1) = 0 
\end{equation}
then the threshold policy that sets $r(t) = F^{-1}(y(t))$ is a
  $(1-\eps)^3 \beta$-approximation to the prophet benchmark, i.e. $\OPT \geq \ALG \geq (1-\eps)^3 \beta \cdot
  \MAX$.
\end{lemma}

\begin{proof}
  Since $r(t)$ goes from the top to the bottom of the support of $F$ when we
  vary $t$ from $0$ to $1$, we can apply the change of variables $r = r(t)$ in
  the prophet bound in Lemma \ref{lemma:prophet_bound} obtaining:
  \begin{equation}\label{eq:prophet_rewritten} 
    \MAX = \int_{0}^1 R(r(t)) e^{-w(r(t))} w'(r(t)) r'(t) dt
  \end{equation}
Now, substituting $y(t) = e^{-w(r(t))}$ and $\hat{R}(t) = R(r(t))$, we get:
  \begin{equation}\label{eq:compare_bounds}  
  \begin{aligned}
   & \MAX & = & \int_{0}^1 \hat{R}(t) [-y'(t)] dt \\
   & \ALG & \geq  (1-\eps)^3 & \int_{0}^1 \hat{R}(t)  \exp\left(  \int_0^t \log y(s) ds \right) dt
  \end{aligned}
  \end{equation}
  The conditions in equation \eqref{eq:prob_a} together with the expression above
  directly imply that the policy is a $(1-\eps^3)\beta$-approximation.
\end{proof}

The final step is to find a function satisfying the conditions in equation
\eqref{eq:prob_a} with $\beta$ equal to the the Kertz bound. It turns out that
that equation \eqref{eq:prob_a} can be transformed in the equation defining the
Kertz bound in equation \eqref{eq:kertz}. The proof of the following lemma can be found in the appendix.

\begin{lemma}\label{lemma:differential_equation}
  There is a solution to equation \eqref{eq:prob_a} with $\beta$ equal to the
  Kertz bound.
\end{lemma}

Taken together, the previous lemmas imply a proof to Theorem \ref{thm:small_prophets}.

\subsection{Constructing the worst case distribution}

The machinery developed in this section also allows us to derive a worst case distribution for i.i.d. prophets in a simple way. As a consequence of Lemma \ref{lemma:kertz_small}, the worst case instance occurs for small variables. In particular, we can simply specify the c.d.f. $F$ of the $\max_i X_i$ and then the instance will be composed of $n$ variables with c.d.f. $F^{1/n}$. We will study the case where $n \rightarrow \infty$ and $\eps \rightarrow 0$.

For the lower bound (positive result) in Lemma \ref{lemma:yFr} we showed that the policy $r(t) = F^{-1}(y(t))$ is a $\beta$-competitive with respect to the prophet. To show it is tight, re-deriving Kertz original upper bound, we reverse engineeer a distribution $F$ such that the optimal policy is forced to have this form. 

We will define a distribution $F_q$ parametrized by $q \in (0,1)$ and obtain the worst case instance in the limit as $q \rightarrow 0$. It is convenient to define $p \in (0,1)$ such that $p = y(q)$ (recall that $y(\cdot)$ is the solution to the differential equation \eqref{eq:prob_a}). Let's also define the following constant (which should be thought as a very large number as $p \rightarrow 1$):
$$H = \frac{1}{y'(q) \log p} - \int_q^1 \frac{1}{y'(t)} dt $$
Next,  we define a function $r^*:[0,1] \rightarrow [0,H]$ 
\begin{equation}\label{eq:optimal_r}
r^*(t) = -\int_t^1 \frac{1}{y'(s)} ds \text{ for } t \in [q,1] \text{ and } r^*(t) = H \text{ for } t \in [0,q]
\end{equation}
With it, we define the c.d.f. $F_q$ as follows:
$$F_q(x) = \left\{ \begin{aligned}
& y((r^*)^{-1}(x)) & & \text{for } t \in [0, r^*(q)] \\
& p & & \text{for } t \in (r^*(q), H] \\
& 1 & & \text{for } t \in [H, \infty)
\end{aligned} \right.$$

\begin{theorem}\label{thm:upper_bound}
Consider $n$ i.i.d. variables with distribution $F_q^{1/n}$ for the distribution $F_q$ constructed above. Let $\OPT_q$ denote the optimal policy for this instance in the limit as $n \rightarrow \infty$ and let $\MAX_q$ denote the prophet benchmark (which is independent of $n$). Then $\lim_{q \rightarrow 0} \OPT_q / \MAX_q = \beta$.
\end{theorem}

The proof can be found in Appendix \ref{apx:proof_thm_upper_bound}.

\newcommand{\ycurve}{(0.000000,1.000000)--(0.010000,0.986585)--(0.020000,0.973173)--(0.030000,0.959763)--(0.040000,0.946360)--(0.050000,0.932964)--(0.060000,0.919577)--(0.070000,0.906201)--(0.080000,0.892838)--(0.090000,0.879491)--(0.100000,0.866161)--(0.110000,0.852850)--(0.120000,0.839560)--(0.130000,0.826293)--(0.140000,0.813051)--(0.150000,0.799837)--(0.160000,0.786653)--(0.170000,0.773500)--(0.180000,0.760381)--(0.190000,0.747297)--(0.200000,0.734252)--(0.210000,0.721247)--(0.220000,0.708284)--(0.230000,0.695366)--(0.240000,0.682495)--(0.250000,0.669673)--(0.260000,0.656902)--(0.270000,0.644185)--(0.280000,0.631524)--(0.290000,0.618920)--(0.300000,0.606377)--(0.310000,0.593897)--(0.320000,0.581481)--(0.330000,0.569133)--(0.340000,0.556854)--(0.350000,0.544647)--(0.360000,0.532513)--(0.370000,0.520456)--(0.380000,0.508477)--(0.390000,0.496579)--(0.400000,0.484765)--(0.410000,0.473035)--(0.420000,0.461393)--(0.430000,0.449840)--(0.440000,0.438380)--(0.450000,0.427013)--(0.460000,0.415743)--(0.470000,0.404571)--(0.480000,0.393500)--(0.490000,0.382532)--(0.500000,0.371669)--(0.510000,0.360912)--(0.520000,0.350265)--(0.530000,0.339729)--(0.540000,0.329306)--(0.550000,0.318998)--(0.560000,0.308807)--(0.570000,0.298736)--(0.580000,0.288785)--(0.590000,0.278957)--(0.600000,0.269254)--(0.610000,0.259677)--(0.620000,0.250228)--(0.630000,0.240910)--(0.640000,0.231723)--(0.650000,0.222670)--(0.660000,0.213751)--(0.670000,0.204969)--(0.680000,0.196325)--(0.690000,0.187821)--(0.700000,0.179458)--(0.710000,0.171237)--(0.720000,0.163160)--(0.730000,0.155228)--(0.740000,0.147443)--(0.750000,0.139805)--(0.760000,0.132316)--(0.770000,0.124977)--(0.780000,0.117789)--(0.790000,0.110753)--(0.800000,0.103870)--(0.810000,0.097142)--(0.820000,0.090568)--(0.830000,0.084150)--(0.840000,0.077890)--(0.850000,0.071787)--(0.860000,0.065843)--(0.870000,0.060058)--(0.880000,0.054434)--(0.890000,0.048971)--(0.900000,0.043671)--(0.910000,0.038534)--(0.920000,0.033561)--(0.930000,0.028755)--(0.940000,0.024116)--(0.950000,0.019647)--(0.960000,0.015350)--(0.970000,0.011229)--(0.980000,0.007290)--(0.990000,0.003540)}

\newcommand{\rcurve}{(0,1)--(.2,1)--(0.200000,0.487117)--(0.210000,0.485782)--(0.220000,0.484380)--(0.230000,0.482908)--(0.240000,0.481368)--(0.250000,0.479758)--(0.260000,0.478077)--(0.270000,0.476326)--(0.280000,0.474503)--(0.290000,0.472608)--(0.300000,0.470640)--(0.310000,0.468598)--(0.320000,0.466481)--(0.330000,0.464290)--(0.340000,0.462021)--(0.350000,0.459676)--(0.360000,0.457253)--(0.370000,0.454750)--(0.380000,0.452167)--(0.390000,0.449503)--(0.400000,0.446756)--(0.410000,0.443925)--(0.420000,0.441008)--(0.430000,0.438006)--(0.440000,0.434915)--(0.450000,0.431734)--(0.460000,0.428463)--(0.470000,0.425098)--(0.480000,0.421639)--(0.490000,0.418084)--(0.500000,0.414430)--(0.510000,0.410676)--(0.520000,0.406820)--(0.530000,0.402858)--(0.540000,0.398790)--(0.550000,0.394612)--(0.560000,0.390322)--(0.570000,0.385917)--(0.580000,0.381395)--(0.590000,0.376752)--(0.600000,0.371985)--(0.610000,0.367091)--(0.620000,0.362067)--(0.630000,0.356908)--(0.640000,0.351611)--(0.650000,0.346172)--(0.660000,0.340586)--(0.670000,0.334850)--(0.680000,0.328957)--(0.690000,0.322904)--(0.700000,0.316684)--(0.710000,0.310292)--(0.720000,0.303722)--(0.730000,0.296967)--(0.740000,0.290021)--(0.750000,0.282877)--(0.760000,0.275526)--(0.770000,0.267960)--(0.780000,0.260172)--(0.790000,0.252150)--(0.800000,0.243885)--(0.810000,0.235367)--(0.820000,0.226583)--(0.830000,0.217521)--(0.840000,0.208167)--(0.850000,0.198507)--(0.860000,0.188524)--(0.870000,0.178201)--(0.880000,0.167519)--(0.890000,0.156457)--(0.900000,0.144991)--(0.910000,0.133095)--(0.920000,0.120741)--(0.930000,0.107894)--(0.940000,0.094519)--(0.950000,0.080570)--(0.960000,0.065996)--(0.970000,0.050734)--(0.980000,0.034704)--(0.990000,0.017794)}

\newcommand{\fcurve}{(0.000000,0.000000)--(0.010000,0.001178)--(0.020000,0.002424)--(0.030000,0.003723)--(0.040000,0.005077)--(0.050000,0.006485)--(0.060000,0.007947)--(0.070000,0.009464)--(0.080000,0.011038)--(0.090000,0.012670)--(0.100000,0.014360)--(0.110000,0.016111)--(0.120000,0.017923)--(0.130000,0.019799)--(0.140000,0.021740)--(0.150000,0.023747)--(0.160000,0.025823)--(0.170000,0.027970)--(0.180000,0.030190)--(0.190000,0.032484)--(0.200000,0.034855)--(0.210000,0.037305)--(0.220000,0.039836)--(0.230000,0.042451)--(0.240000,0.045153)--(0.250000,0.047943)--(0.260000,0.050824)--(0.270000,0.053799)--(0.280000,0.056871)--(0.290000,0.060042)--(0.300000,0.063316)--(0.310000,0.066695)--(0.320000,0.070183)--(0.330000,0.073782)--(0.340000,0.077495)--(0.350000,0.081327)--(0.360000,0.085280)--(0.370000,0.089357)--(0.380000,0.093563)--(0.390000,0.097900)--(0.400000,0.102372)--(0.410000,0.106983)--(0.420000,0.111737)--(0.430000,0.116637)--(0.440000,0.121687)--(0.450000,0.126891)--(0.460000,0.132252)--(0.470000,0.137776)--(0.480000,0.143464)--(0.490000,0.149323)--(0.500000,0.155354)--(0.510000,0.161563)--(0.520000,0.167954)--(0.530000,0.174530)--(0.540000,0.181295)--(0.550000,0.188253)--(0.560000,0.195408)--(0.570000,0.202764)--(0.580000,0.210325)--(0.590000,0.218093)--(0.600000,0.226074)--(0.610000,0.234269)--(0.620000,0.242684)--(0.630000,0.251320)--(0.640000,0.260180)--(0.650000,0.269269)--(0.660000,0.278588)--(0.670000,0.288141)--(0.680000,0.297929)--(0.690000,0.307954)--(0.700000,0.318219)--(0.710000,0.328725)--(0.720000,0.339474)--(0.730000,0.350466)--(0.740000,0.361703)--(0.750000,0.373185)--(0.760000,0.384912)--(0.770000,0.396884)--(0.780000,0.409100)--(0.790000,0.421561)--(0.800000,0.434264)--(0.810000,0.447208)--(0.820000,0.460391)--(0.830000,0.473811)--(0.840000,0.487465)--(0.850000,0.501351)--(0.860000,0.515463)--(0.870000,0.529800)--(0.880000,0.544356)--(0.890000,0.559127)--(0.900000,0.574108)--(0.910000,0.589293)--(0.920000,0.604678)--(0.930000,0.620255)--(0.940000,0.636019)--(0.950000,0.651962)--(0.960000,0.668078)--(0.970000,0.684360)--(0.980000,0.700800)--(0.990000,0.717390)--(1.9,0.717390)--(1.9,1)--(2,1)}

\begin{figure}[t]
\centering
\begin{tikzpicture}[xscale=3, yscale=3]
\draw[line width=1,  color=blue] \ycurve;
\draw[->] (-0.05,0) -- (1.05,0);
\draw[->] (0.0,-0.05) -- (0.0,1.05);
\draw (1,-.02)--(1,.02);
\draw (-.02,1)--(.02,1);
\node at (1,-.1) {$1$};
\node at (-.06,1) {$1$};
\node at (1.05,.1) {$t$};
\node at (.2,1) {$y(t)$};

\begin{scope}[xshift=40]
\draw[->] (-0.05,0) -- (1.05,0);
\draw[->] (0.0,-0.05) -- (0.0,1.05);
\draw[line width=1,  color=blue] \rcurve;
\node at (-.2,.5) {$r^*(q)$};
\node at (.4,1) {$r^*(t)$};
\node at (-.1,1) {$H$};
\draw (1,-.02)--(1,.02);
\draw (.2,-.02)--(.2,.02);
\node at (.2,-.1) {$q$};
\node at (1,-.1) {$1$};
\draw (-.02,.5)--(.02,.5);
\draw (-.02,1)--(.02,1);
\end{scope}

\begin{scope}[xshift=80]
\draw[->] (-0.05,0) -- (2.05,0);
\draw[->] (0.0,-0.05) -- (0.0,1.05);
\draw[line width=1,  color=blue] \fcurve;
\draw (0.99,-.02)--(0.99,.02);
\draw (1.9,-.02)--(1.9,.02);
\node at (0.99,-.1) {$r^*(q)$};
\node at (1.9,-.1) {$H$};
\draw (-.02,1)--(.02,1);
\draw (-.02,0.717)--(.02,0.717);
\node at (-.1,1) {$1$};
\node at (-.1,0.717) {$p$};
\node at (1.7,1) {$F(x)$};
\end{scope}
\end{tikzpicture}
\caption{The first plot corresponds to the solution $y(t)$ of the differential equation \eqref{eq:prob_a}, which is also equivalent to equations \eqref{eq:prob_b}. The second and third plots correspond to the constructions of $r^*(t)$ and $F_q(x)$.}
\label{fig:y_r_F}
\end{figure}
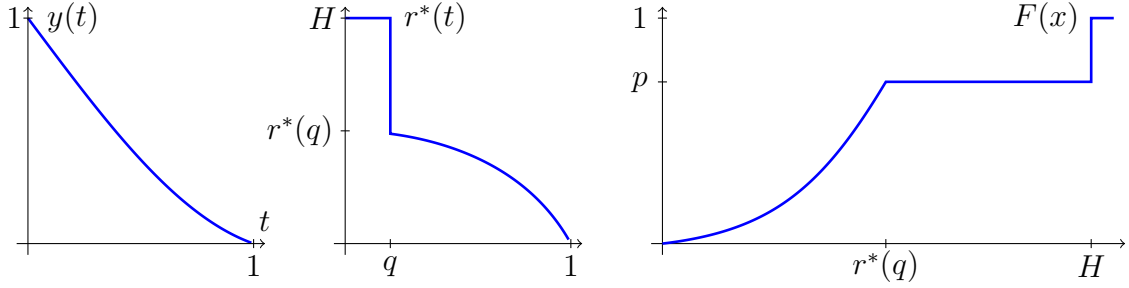

\section{Imperfect Prophets} \label{sect:imperfect}

The Kertz bound $\beta \approx 0.745$ is the golden standard for prophet inequalities since the i.i.d. case establishes a natural upper bound on what can be achieved in any other setting. Our main result in this paper is that \textit{all} prophet instances are \textit{near} optimal: by removing $\poly(\eps^{-1})$ variables from an instance, we can find a stopping time for the resulting instance that achieves a $(\beta - \eps)$ fraction of its maximum.

It is useful to contrast this with the result in \cite{CSZ19}, which shows an instance $\bX$ where the gap between $\OPT(\bX)$ and $\MAX(\bX)$ is strictly less than $\beta$. Their result shows an inherent gap between the i.i.d. and non-i.i.d. cases. We show that by removing a constant number of variables, this gap disappears. In fact, if we compete with a slightly imperfect prophet - who can obtain the $k$-th largest value as reward ($\MAX_k$ in Definition \ref{def:imperfect}) then we are able to achieve the Kertz bound and this is tight. 

\begin{theorem}\label{thm:nearby}
Consider a collection of $n$ random variables $\bX = \{X_1, X_2, \dots, X_n\}$. Then, for any $\eps > 0$, there exists a subset $\bX'$ of $\bX$ containing $n - {\Theta}(\eps^{-3} \log \eps^{-1})$ of these variables so that $$\OPT(\bX') \geq (\beta - O(\eps)) \cdot \MAX(\bX')$$ Moreover, it is possible to efficiently construct such a subset $\bX'$ and a policy achieving this guarantee in polynomial time. 
\end{theorem}

If our stopping rule is allowed to use variables that the prophet is not, then we can get a slightly better dependence on $\varepsilon$. 

\begin{theorem}\label{thm:nearby_weak_intro}
Consider a collection of $n$ random variables $\bX = \{X_1, X_2, \dots, X_n\}$. Then, for any $\eps > 0$, there exists a subset $\bX'$ of $\bX$ containing $n - {\Theta}(\eps^{-2} \log \eps^{-1})$ of these variables so that $\OPT(\bX) \geq (\beta - O(\eps)) \cdot \MAX(\bX')$. Moreover, it is possible to efficiently construct such a subset $\bX'$ and a policy achieving this guarantee in polynomial time. 
\end{theorem}

Theorems \ref{thm:nearby} and \ref{thm:nearby_weak_intro} have a number of almost immediate interesting applications. For example, Theorem \ref{thm:nearby_weak_intro} immediately implies a nearly optimal prophet inequality for the $\Theta(\eps^{-2})$th-order statistic of $\bX$. 

\begin{corollary}\label{cor:kthlargest}
If $\bX$ is any collection of $n$ random variables, then  for any $\eps >0$ there is a $k = \Theta(\eps^{-2} \log \eps^{-1})$ such that
$$\OPT(\bX) \geq (\beta - O(\eps)) \cdot \MAX_k(\bX).$$
\end{corollary}
\begin{proof}
Let $X^{(k)}$ be a random variable denoting the value of the $k$th largest element of $\bX$. Note that for any subset $\bX'$ of $\bX$ with $|\bX'| = n-k$, $\max(\bX') \geq X^{(k)}$. The result then immediately follows from Theorem \ref{thm:nearby_weak_intro}.
\end{proof}

We can also use Theorem \ref{thm:nearby_weak_intro} to design optimal stopping rules for $m$-frequent instances (see Definition \ref{def:m_frequent}) approaching the Kertz bound  as $m$ grows large.

\begin{theorem}\label{thm:frequent}
Consider an $m$-frequent collection $\bX$ of independent random variables with $m = {\Omega}(\eps^{-2} \log \eps^{-1})$ for some $\eps >0$. Then:
$$\OPT(\bX) \geq (\beta - O(\eps)) \cdot \MAX(\bX)$$
\end{theorem}

We will defer the proof of this theorem to  Section \ref{sect:freqextension}.\\

All of the above theorems are tight in the sense that it is impossible to replace the Kertz bound $\beta$ by any larger constant. For frequent prophets (Theorem \ref{thm:frequent}) this follows since i.i.d. prophets are a special case of frequent prophets and the Kertz upper bound holds there. For imperfect prophets, a similar reduction to the i.i.d. case holds. 

\begin{lemma}\label{lem:nearby_tight}
Choose any $\alpha > \beta$ (where $\beta$ is the Kertz bound) and positive integer $r$. Then (for a sufficiently large $n$) there exists a collection of $n$ random variables $\bX = \{X_1, X_2, \dots, X_n\}$ such that for any subset $\bX'$ of $\bX$ containing $n - r$ of these variables,

$$\OPT(\bX) < \alpha \cdot \MAX(\bX').$$ 
\end{lemma}

A proof of Lemma \ref{lem:nearby_tight} can be found in Appendix \ref{app:kertz_tight}.

The remainder of this section is structured as follows. In Section \ref{sect:free_order}, we begin by proving Theorems \ref{thm:nearby} and \ref{thm:nearby_weak_intro} in the weaker model of \textit{free-order} prophets, where the algorithm can choose the order in which it looks at the items (in addition to being a slightly simpler proof, we also get a slightly better dependence on $\eps$). In Section \ref{sect:random_order}, we extend this to random-order prophets and prove Theorems \ref{thm:nearby} and \ref{thm:nearby_weak_intro}. Finally, in Section \ref{sect:freqextension}, we apply these theorems to prove Theorem \ref{thm:frequent} for frequent prophets.

\subsection{Warm-up: free order}\label{sect:free_order}

For simplicity, we will begin by proving Theorem \ref{thm:nearby_weak_intro} in the free-order setting, where the algorithm has the power to choose the order in which they encounter the random variables. 

Recall that this is the weaker analogue of Theorem \ref{thm:nearby}, where the algorithm is allowed to use all the random variables (but the prophet is only allowed to choose the max of a specific subset of all but $\poly(\eps)$ of the variables).

\begin{theorem}\label{thm:nearby_weak}
Consider a collection of $n$ random variables $\bX = \{X_1, X_2, \dots, X_n\}$. Then, for any $\eps > 0$, there exists a subset $\bX'$ of $\bX$ of size $n - \Theta(\eps^{-1}\log \eps^{-1})$ so that $$\OPT_{\free}(\bX) \geq  (\beta - O(\eps)) \cdot \MAX(\bX')$$ Moreover, it is possible to efficiently construct such a subset $\bX'$ and a policy achieving this guarantee in polynomial time. 
\end{theorem}
\begin{proof}
For any real $t \geq 0$, define the random variables $$Y_i(t) = \max(X_i, t) \text{ and } Z_i(t) = Y_i(t) - t$$ Note that as $t$ increases, the number of variables $Z_i(t)$ that are $(\eps, \eps t)$-small increases. In particular, if $Z_i(t)$ is $(\eps, \eps t)$-small for a given $t$, $Z_i(t')$ is $(\eps, \eps t')$ small for all $t' \geq t$. Also, note that every variable $Z_i(t)$ is $(\eps, \eps t)$-small for a sufficiently large $t$ (e.g. $t = \max(1, F_{i}^{-1}(1-\eps))$). 

If there are fewer than $k = \Theta\left(\frac{1}{\eps}\log\frac{1}{\eps}\right)$ variables that are $\varepsilon$-small, then we are done. If not, let $t^*$ be the supremum over all $t$ such that at most $k$ of the variables $Z_i(t^*)$ are \textit{not} $(\eps, \eps t^*)$-small. This means that for $k$ values of $i$, $\Pr[Z_i(t^*) \geq \eps t^*] \geq \eps$; in particular, this implies that $\Pr[X_i \geq t^* (1+ \eps)] \geq \eps$. For the remaining indices $Z_i(t^*)$ is $(\eps, \eps t^*)$-small.

Let $\bX_{\bbig}$ be the set of $X_i$ corresponding to these $k$ indices and $\bX' = \bX \setminus \bX_{\bbig}$. Our policy will proceed as follows. We will order the items so that all the elements of $\bX'$ come before the elements of $\bX_{\bbig}$. We will start by running the small prophets policy of Section \ref{sect:small} on the elements in $\bX'$ (specifically, mapping each $X_i$ to $Z_i = \max(X_i - t^*, 0)$ and accepting each $X_i$ with the probability $Z_i$ would have been accepted under the corresponding small prophets instance). If we reach the end of $\bX'$ without accepting any items, we accept the first item we see in $\bX_{\bbig}$ that is larger than $t^*$. 

Let us analyze the performance of this policy. First, we claim that with high probability, $\bX_{\bbig}$ contains an $X_i$ greater than or equal to $t^*$. Recall that each $X_i$ in $\bX_{\bbig}$ satisfies $\Pr[X_i \geq t^*(1 + \eps)] \geq \eps$. Therefore, 
\begin{eqnarray*}
\Pr[\max(\bX_{\bbig}) \leq t^*] &\leq& (1-\eps)^{|\bX_{\bbig}|} = (1-\eps)^{k} = (1-\eps)^{\Theta(\eps^{-1}\log\eps^{-1})} 
\leq \eps.
\end{eqnarray*}

It follows that with probability at least $(1-\eps)$, $\bX_{\bbig}$ contains an $X_i$ satisfying $X_i \geq t^*$. In particular, with probability at least $(1-\eps)$, our policy is guaranteed to receive reward at least $t^*$ (since we also only accept an $X_i \in \bX'$ if $X_i > t^*$). 

Now, note that the only way our algorithm ends up picking something in $\bX_{\bbig}$ is if the small prophets subpolicy did not pick any item in $\bZ'$ (achieving a score of 0). In this case (conditioned on $\bX_{\bbig}$ containing an $X_i$ satisfying $X_i \geq t^*$) our algorithm receives a reward of at least $t^*$. On the other hand, if the small prophets subpolicy does pick an item $Z_i$ in $\bZ'$, the algorithm receives a reward of $X_i = Z_i + t^*$ (note that the small prophets subpolicy will never select a $Z_i$ equal to 0). 

Conditioning on $\bX_{\bbig}$ containing a big element (which happens with probability at least $1-\eps$), this means that the expected reward of our algorithm is at least $t^*$ plus the expected reward of the small prophets subpolicy. By Corollary \ref{cor:eps-delta-small}, running the small prophets policy on $\bZ'$ guarantees us an additional $(\beta- O(\eps))\cdot \MAX(\bZ') - \eps t^*$ reward over our guaranteed $t^*$. The total expected reward of our algorithm is therefore at least

\begin{eqnarray*}
\OPT(\bX) &\geq & (1-\eps)((1 - O(\eps))\beta \cdot \MAX(\bZ') - \eps t^* + t^*) \\
&\geq & (\beta-O(\eps))\cdot (\MAX(\bZ') + t^*) \\
&= & (\beta-O(\eps)) \cdot \MAX(\bX').
\end{eqnarray*}

\end{proof}

We now extend this to the setting where the algorithm is also restricted to the same subinstance as the prophet. To do this, we roughly proceed as follows. If the expected maximum of the subinstance is close to the expected maximum of the original instance (i.e. $\MAX(\bX') \geq (1-\eps)\MAX(X)$), then we are done -- we can just take our final subinstance to be $\bX$. If it is not, then we can recurse and apply Theorem \ref{thm:nearby_weak} to $\bX'$: if it has a subinstance $\bX''$ with $\MAX(\bX'') \geq (1-\eps)\MAX(\bX')$ then we are also done (we can take our final subinstance to be $\bX'$). Otherwise, we continue recursing.

We repeat until we have done this $\poly(1/\eps)$ times. At this point we are looking at some subinstance $\bX^*$ of $\bX$ (with $|\bX^*| = |\bX| - \poly(1/\eps)$). Since we have never stopped and reported a valid subinstance, we know that $\MAX(\bX^*) \leq (1-\eps)^{\poly(1/\eps)}\MAX(\bX) \leq \poly(\eps)\MAX(\bX)$. But in this case we can show that one of the variables we discarded has expectation much larger than $\MAX(\bX^*)$; by including it in our subinstance and just picking it, we can achieve a competitive ratio very close to 1. We formalize this in the following theorem.

\begin{theorem}\label{thm:nearby_strong}
Consider a collection of $n$ random variables $\bX = \{X_1, X_2, \dots, X_n\}$. Then, for any $\eps > 0$, there exists a subset $\bX'$ of $\bX$ of size $n - \Theta((\eps^{-1}\log \eps^{-1})^2)$ so that $$\OPT_{\free}(\bX') \geq (\beta - O(\eps))\cdot \MAX(\bX')$$ Moreover, it is possible to efficiently construct such a subset $\bX'$ and a policy achieving this guarantee in polynomial time. 
\end{theorem}

\begin{proof}
We will define a sequence of instances in the following way. Let $\bX_0 = \bX$ be the original instance, and for each $i$, let $\bX_{i+1}$ be the subset of $\bX_i$ constructed by Theorem \ref{thm:nearby_weak} when applied to $\bX_i$. From the conditions of Theorem \ref{thm:nearby_weak}, for each $k \geq 0$ we know that $|\bX_k| \geq n - O(k\eps^{-1}\log\eps^{-1})$ and that

$$\OPT_{\free}(\bX_{k}) \geq (\beta - O(\eps)) \cdot \MAX(\bX_{k+1}).$$

Now, if there exists a $k < k_{max} = 10\eps^{-1}\log\eps^{-1}$ such that $\MAX(\bX_{k+1}) \geq (1 - \eps)\MAX(\bX_{k})$, this immediately implies the desired result (for $\bX' = \bX_{k}$). Thus, assume that for all such $k$, $\MAX(\bX_{k+1}) < (1-\eps)\MAX(\bX_{k})$. 

This implies that, for all $k>0$, $\MAX(\bX_{k}) < (1-\eps)^k \MAX(\bX_0)$. In particular, for $k = 1$, we know that $\MAX(\bX_{1}) < (1-\eps)\MAX(\bX)$. We'll now argue that there exists a random variable $X^* \in \bX \setminus \bX_{1}$ such that $\E[X^*] \geq \poly(\eps)\MAX(\bX)$. We will then argue that we can construct a good policy for $\bX' = \bX_{k_{max}} \cup \{X^*\}$ (in particular, $X^*$ will have such large expectation compared to $\MAX(\bX_{k_{max}})$ that it will suffice to just accept $X^*$).

Let $\bX_{\bbig} = \bX \setminus \bX_{1}$. Recall that for any non-negative random variables $U$ and $V$ that $\E[U] \geq \E[\max(U, V)] - \E[V]$ (since $U + V \geq \max(U, V)$). Applying this to $U = \bX_{\bbig}$ and $V = \bX_{1}$, we have that

$$\E[\max(\bX_{\bbig})] \geq \MAX(\bX) - \MAX(\bX_1) > \eps \MAX(\bX).$$

Since $|\bX_{\bbig}| \leq \Theta(\eps^{-1}\log \eps^{-1})$, this implies that there exists a variable $X^* \in \bX_{\bbig}$ such that $\E[X^*] > \Theta(\eps^{2}/\log\eps^{-1})\MAX(\bX)$. 

Now, consider the subset of variables $\bX' = \bX_{k_{max}} \cup \{X^*\}$. Note that

\begin{eqnarray*}
\frac{\MAX(\bX')}{\OPT_{\free}(\bX')} &\leq & \frac{(1-\eps)^{k_{max}}\MAX(\bX) + \E[X^*]}{\E[X^*]} \\
&=& 1 + \frac{(1-\eps)^{k_{max}}\MAX(\bX)}{\E[X^*]} \\
&\leq & 1 + \frac{(1-\eps)^{k_{max}}}{\Theta(\eps^{2}/\log\eps^{-1})} \\
& \leq & 1 + \frac{\exp(-\eps \cdot (10\eps^{-1}\log\eps^{-1}))}{\Theta(\eps^{2}/\log\eps^{-1})} \\
& \leq & 1 + \frac{\exp(-10\log\eps^{-1}))}{\Theta(\eps^{2}/\log\eps^{-1})} \\
& \leq & 1 + O(\eps^{8}\log\eps^{-1}).
\end{eqnarray*}

It follows that $\OPT_{\free}(\bX') \geq (1 - O(\eps)) \cdot \MAX(\bX) \geq (\beta - O(\eps))\cdot \MAX(\bX')$, as desired.
\end{proof}

\subsection{Random order}\label{sect:random_order}

We now prove Theorems \ref{thm:nearby} and \ref{thm:nearby_weak_intro} by extending the proofs of Theorems  \ref{thm:nearby_weak} and \ref{thm:nearby_strong} to the random-order model.
The only place in Section \ref{sect:free_order} where we use our ability to control the order of the variables is in Theorem \ref{thm:nearby_weak}, where we place all the items in $\bX_{\bbig}$ after those in $\bX'$; other than this, everything works given the variables in a random order. Instead of ordering the items in $\bX_{\bbig}$ after those in $\bX'$, we'll instead show that enough of the variables in $\bX_{\bbig}$ will occur late enough (e.g. after we have seen $(1 - \poly(\eps))$ of the items) that we can replicate the proof with these items.\\

To do this, we will first need to slightly strengthen our policy for small prophets so that it can work even when restricted to select an item from the first $(1-\eps)$ fraction of random variables.

\begin{lemma}\label{lem:modifiedsmall}
Let $\bX$ be a collection of $n$ $(\eps, \delta)$-small variables. Then, if $n > \eps^{-2} \log (1/\eps)$, there exists a stopping time $\tau$ such that $\tau \leq n(1-\eps)$ (the policy always stops before element $(1-\eps) n$) and such that

$$\E[X_{\tau}] \geq (\beta - O(\eps)) \cdot \MAX(\bX) - \delta.$$
\end{lemma}
\begin{proof}
Consider a slightly different policy, which begins by sampling a positive integer $R$ from $\mathrm{Binomial}(n, (1-\eps))$ and guarantees that its stopping time $\tau$ is at most $R$. We will show that we can construct such a policy which is $(\beta - O(\eps))$ competitive with $\MAX(\bX)$. 

To do this, consider the collection $\bX'$ of variables defined via $X'_i = X_i$ with probability $1 - \eps$, and $X'_i = 0$ with probability $0$. Note that $\bX'$ is also a collection of $(\eps, \delta)$-small variables. We will show how to transform the small prophets policy for $\bX'$ to a small prophets policy for $\bX$ that stops within the first $R$ elements.

Recall that the small prophets policy for $\bX'$ can be implemented as follows. We start by sampling $n$ random times uniformly from the interval $[0, 1]$; in sorted order, let these times be $t'_1 \leq t'_2 \leq \dots \leq t'_n$. The small prophets policy $\tau'$ provides a threshold function $r(t)$ such that we should select the $i$th item $X'_{\pi(i)}$ we encounter if $X'_{\pi(i)} \geq r(t'_i)$ (and we have not selected any earlier item). Importantly, note that the $i$th threshold does not depend on the individual identity of the $i$th item, just on the randomly sampled $t_i$ and this global threshold function $r$. By Corollary \ref{cor:eps-delta-small}, this guarantees an expected reward at least $\E[X'_{\tau'}] \geq (\beta - O(\eps)) \cdot \MAX(\bX') - \delta$.

Now, consider the policy $\tau$ for $\bX$ where we sample $R$ random times uniformly from $[0, 1]$ (call them $t_1 \leq t_2 \leq \dots \leq t_R$), and select the $i$th item we encounter if $X_{\pi(i)} \geq r(t_i)$ (if we pick no item by $X_{\pi(r)}$, we end the protocol without picking anything). We claim that this policy gets exactly the same reward in expectation as $\tau'$; i.e. that $\E[X_{\tau}] = \E[X'_{\tau'}]$. To see this, note that we can couple executions of $\tau'$ on $\bX'$ with executions of $\tau$ on $\bX$. Specifically, $R$ should equal the number of indices $i$ where $X'_{i}$ is chosen to equal $X_{i}$ (instead of $0$), and $t_i$ should equal the $i$th value of $t'_j$ that corresponds to a ``non-zero'' $X_{i}$ (it is straightforward to verify that the distributions of $R$ and $t_i$ match the distributions generated by this coupling process). 

It follows that $\E[X_{\tau}] = \E[X'_{\tau'}] \geq (\beta - O(\eps)) \cdot \MAX(\bX') - \delta$. Now, note that $\MAX(\bX') \geq (1-\eps)\MAX(\bX)$ (one way to see this is to note that with probability $(1-\eps)$, the maximum item in $\bX$ does not get erased and remains the maximum in $\bX'$). It follows that $\E[X_{\tau}] \geq (\beta - O(\eps)) \cdot \MAX(\bX') - \delta$.

Finally, all the previous discussion was for a modified stopping time which observes the first $\mathrm{Binomial}(n, (1-\eps))$ items in $\bX$. But to prove the theorem statement, we wish to show that there exists a stopping time which never observes more than $(1-\eps)n$ items. To do this, we will apply the first result for $\eps' = 2\eps$ so that we only observe the first $R' = \mathrm{Binomial}(n, (1-\eps'))$ items (since $\eps' = O(\eps)$, our guarantee is still the same). Then note that by Hoeffding's inequality,
the probability $R'$ is larger than $(1-\eps)n$ is at most $\exp(-2\eps^{2} n) \leq \eps^{2}$ (since $n > \eps^{-2}\log 1/\eps$). On the other hand, if $R' > (1-\eps)n$, we lose reward at most $\eps^{2}\MAX(\bX)$ (since the choice of $R'$ is independent from the realizations of $\bX$). This means that if we follow the stopping time for $R'$ and quit if $R' > (1-\eps)n$, our stopping time satisfies

$$\E[X_{\tau}] \geq (\beta - O(\eps)) \MAX(\bX) - \delta - \eps^{2}\MAX(\bX)  \geq (\beta - O(\eps)) \MAX(\bX) - \delta.$$

\end{proof}

We can now proceed to prove Theorem \ref{thm:nearby_weak_intro}.

\begin{proof}[Proof of Theorem \ref{thm:nearby_weak_intro}]
Similarly to the proof of Theorem \ref{thm:nearby_weak}, set $k = \Theta(\eps^{-2}\log \eps^{-1})$, and let $t^*$ be the supremum over all $t$ such that at most $k$ of the variables $Z_{i}(t^*)$ are not $(\eps, \eps t^*)$-small. Let $\bX_{\bbig}$ be this subset of $k$ non-small variables, and let $\bX' = \bX \setminus \bX_{\bbig}$. 

Our policy will operate in two parts. We will first use the small prophets policy on the variables $Z_i$ to process elements of $\bX'$ (ignoring elements of $\bX_{\bbig}$) until we have seen $(1-\eps)|\bX'|$ elements of $\bX'$. If we have not chosen an item by this time, we will pick the next item we see with value greater than or equal to $t^*$. 

As in the proof of Theorem \ref{thm:nearby_weak}, any item we choose is guaranteed to have value at least $t^*$. Let $\bX_{\rand}$ be the set of $(1-\eps)|\bX'|$ elements we process in the first part of our policy. By Lemma \ref{lem:modifiedsmall}, the modified small prophets policy on $\bX_{\rand}$ guarantees us an additional reward of $(\beta - O(\eps)) \cdot \MAX(\bX') - \eps t^*$ over $t^*$. It therefore follows that, conditioned on our policy picking an item, our policy achieves expected reward at least 

$$(\beta - O(\eps)) \cdot \MAX(\bZ_{\rand}) - \eps t^* + t^* \geq (\beta - O(\eps)) \cdot \MAX(\bX_{\rand}) \geq (\beta - O(\eps)) \cdot \MAX(\bX').$$

We now claim that our policy selects an item with probability at least $(1 - O(\eps))$, thus implying the theorem. First, note that each $X_i \in \bX_{\bbig}$ has an independent $\eps$ probability of occurring after the $(1-\eps)|\bX'|$th item in $\bX'$ (to see this, imagine constructing the uniform random order by first choosing the order of elements in $\bX'$ and then randomly inserting the elements in $\bX_{\bbig}$). It follows from Hoeffding's inequality that the probability we see at least $k\eps/2$ variables in $\bX_{\bbig}$ in the second part of the policy is at least

$$1 - 2\exp(-\eps^2k/2) = 1 - 2\exp(-\Theta(\log \eps^{-1})) \geq 1 - O(\eps).$$

Since each of the variables $X_i \in \bX_{\bbig}$ satisfies $\Pr[X_i \geq t^*(1+\eps)] \geq \eps$, the probability that maximum of these $k\eps/2$ variables is less than $t^*$ is at most $(1 - \eps)^{k\eps/2} = (1 - \eps)^{\Theta(\eps^{-1}\log\eps^{-1})} \leq \varepsilon$. It follows that with probability at least $1 - O(\varepsilon)$ that in the second part of our policy we see a variable with value at least $t^*$, as desired. 
\end{proof}

\begin{proof}[Proof of Theorem \ref{thm:nearby}]
Follows from the proof of Theorem \ref{thm:nearby_strong}, using Theorem \ref{thm:nearby_weak_intro} in place of Theorem \ref{thm:nearby_weak}.
\end{proof}

\subsection{Frequent Prophets}\label{sect:freqextension}

Finally, we show the proof of Theorem \ref{thm:frequent} as an application of Theorem \ref{thm:nearby_weak_intro}. The following lemma relates the expected max of a subinstance of a frequent instance to the expected max of the instance itself.

\begin{lemma}\label{lem:freqcoupling}
Let $\bX$ be an $m$-frequent collection of random variables. Let $\bX' \subseteq \bX$ have size $|\bX'| \geq |\bX| - k$. Then
$$\MAX(\bX') \geq \left(1 - \frac{k}{m}\right)\MAX(\bX).$$
\end{lemma}

\begin{proof}
We will show that with probability $1 - \frac{k}{m}$, $\max(\bX') = \max(\bX)$, hence implying the result.

Recall that in an $m$-frequent collection of random variables, each random variable is distributed according to some distribution $\D_i$, and each distribution $\D_i$ has at least $m$ random variables distributed according to it. Let us condition on the event that the maximum variable in $\bX$ is distributed according to $\D_i$. By symmetry, any of the (at least) $m$ random variables distributed according to $\D_i$ has an equal chance of being the maximum. Moreover, at least $m-k$ of these variables also belong to $\bX'$. It follows that the probability that the maximum variable belongs to $\bX'$ (conditioned on the variable being distributed according to $\D_i$) is at least $1 - \frac{k}{m}$. Since this is true for each $\D_i$, it follows (by the law of total probability) that it is true in general, as desired.
\end{proof}

The result for frequent prophets follows as a direct corollary.

\begin{proof}[Proof of Theorem \ref{thm:frequent}]
Follows from Lemma \ref{lem:freqcoupling} and Theorem \ref{thm:nearby_weak_intro}.
\end{proof}

\section{An Efficient PTAS for Optimal Ordering} \label{sect:optimal_order}

Finally, we turn our attention to the problem of efficiently constructing good policies for the free order stopping problem. Recall that an instance of this problem consists of $n$ nonnegative random variables $X_1,X_2, \dots, X_n$ with known distributions. The learner can choose to view the variables in any order but after observing $X_i$, the learner can either choose to stop and accept the reward given by $X_i$ or discard $X_i$ and observe the next variable. Our goal is to design a policy for the learner which maximizes their expected reward.  Note that once the order of the observed variables is fixed, the optimal policy can be easily computed by solving a dynamic program\footnote{In particular, one should always accept the $i$th random variable in the ordering iff it is above the threshold given by the maximum expected reward you can obtain from the remaining $n-i$ variables.}. Thus, the problem of constructing an optimal policy reduces to finding an optimal ordering.

In this section we apply our decomposition technique to provide efficient algorithms for constructing approximately optimal orderings. We say an ordering is \emph{$\eps$-optimal} if the expected reward when following the (policy corresponding to this) ordering is at least $(1-\eps)\OPT_{\free}$. Note that, unlike in the previous sections, here we are not concerned with competing against the prophet benchmark ($\E[\max X_i]$) but instead the performance of the true optimal ordering  ($\OPT_{\free}$). We begin by providing an efficient algorithm for computing an $\eps$-optimal ordering when all variables are $\eps$-small.

\begin{theorem}\label{thm:small_ordering}
Let $\bX = \{X_1, X_2, \dots, X_n\}$ be a collection of $\eps$-small random variables. There is an algorithm which runs in time $\poly(n, \eps^{-1})$ and computes an $O(\eps)$-optimal ordering for this collection of random variables. 
\end{theorem}

We then leverage this algorithm to design an efficient polynomial-time approximation scheme (EPTAS) for computing an $\eps$-optimal ordering for \emph{any} collection of random variables.

\begin{theorem}\label{thm:general_ordering}
Let $\bX = \{X_1, X_2, \dots, X_n\}$ be a collection of non-negative random variables. There is an algorithm which runs in time $\exp(\tilde{O}(\eps^{-2}))\poly(n)$ and computes an $\eps$-optimal ordering for this collection of random variables.
\end{theorem}

Note that all previous algorithms for this problem had arbitrarily bad polynomial dependence on $n$ (e.g. time complexities that scale as $n^{1/\eps}$). This is the first known PTAS for this problem which is ``efficient'' (i.e., whose complexity can be written in the form $O(f(\eps)\poly(n))$).

The remainder of this section is organized as follows. In Section \ref{sec:structure}, we begin by making some observations about the structure of the optimal policy for the free order stopping problem. In Section \ref{sec:small_ordering}, we show how to write a convex program to find a nearly optimal ordering for $\eps$-small random variables, thus establishing Theorem \ref{thm:small_ordering}. Finally, in Section \ref{sec:general_ordering} we show how to combine this with our decomposition technique, establishing Theorem \ref{thm:general_ordering}.

\subsection{Structural observations}\label{sec:structure}

We begin with some structural observations on the form of an optimal policy for the optimal stopping problem. In particular, we'll show that any optimal policy can be well approximated by a policy that sets a decreasing threshold discretized to a small set of values. 

To start, we will narrow our attention to \textit{stateless policies}, which can be completely specified by an ordering $\pi:[n]\rightarrow [n]$ of the $n$ random variables $X_i$ and a threshold $\tau:[n] \rightarrow \R_{\geq 0}$ for each item. To execute such a policy, we visit the random variables in the order $X_{\pi(1)}, X_{\pi(2)}, \dots, X_{\pi(n)}$. When we visit $X_{\pi(i)}$, we check whether its instantiation is greater than the threshold $\tau(\pi(i))$. If so, we immediately accept (ending the protocol). Otherwise, we continue to the next item. Without loss of generality, $\tau(\pi(n)) = 0$, since it is always beneficial to accept the very last item. 

Fix a policy $\Pi = (\pi, \tau)$ and define:
$$\lambda_i = \E[X_{\pi(i)} | X_{\pi(i)} \geq \tau(\pi(i))] \quad \text{and} \quad p_i = \Pr[X_{\pi(i)} < \tau(\pi(i))]$$
which correspond to the expected value of the $i$th random variable we see conditioned on it exceeding its threshold and the  probability that the $i$th random variable we see falls short of its threshold. From the description of the protocol, we can see that the expected reward $U(\Pi)$ of this policy $\Pi$ can be written as:

\begin{equation}\label{eq:utilpi}
    U(\Pi) = \sum_{i=1}^{n}\lambda_i (1-p_i)\prod_{j < i} p_j.
\end{equation}

Via the dynamic programming argument mentioned earlier, we can show the optimal policy belongs to this class of stateless policies.

\begin{lemma}\label{lem:opt_policy}
The optimal policy is a stateless policy that satisfies

\begin{equation}\label{eq:opt_tau}
\tau(k) = \sum_{i=k+1}^{n} \lambda_i (1-p_i)\left( \prod_{(k+1) \leq j < i} p_j\right).
\end{equation}
\end{lemma}

The proof (which can be found in the Appendix \ref{apx:optimal_ordering}) follows from the usual argument: the optimal policy chooses a variable iff its value is larger than the expected reward it can obtain from the remaining variables.

As a consequence of this, we note that if an optimal policy ever sees a sufficiently large variable, it must immediately accept it (regardless of position). More concretely, we can bound the maximum value of any threshold in the optimal policy for $\bX$.
\begin{lemma}\label{lemma:max_threshold}
The optimal policy immediately accepts any variable whose realization is at least $\MAX(\bX)$.
\end{lemma}
\begin{proof}
The right hand side of equation \eqref{eq:opt_tau} is always at most $\OPT_{\free}(\bX)$, which in turn is at most $\MAX(\bX)$.
\end{proof}

We next show that we can discretize the set of possible thresholds without sustaining much loss in expected reward.
\begin{lemma}\label{lem:discretizedthresholds}
There exists a policy for $\bX$ that achieves a reward of at least $(1 - 2\eps)\OPT_{\free}(\bX)$ only using thresholds of the form $t_u = (1-u\eps)^+ \MAX$ for $u \in \{0, 1, 2, \hdots, \lceil 1/\eps \rceil\}$.
\end{lemma}

The proof (in Appendix \ref{apx:optimal_ordering}) follows from rounding all the thresholds in the optimal policy down to the closest multiple of $\eps$ and keeping track of the change in the objective value.

Finally, we show that given an assignment of thresholds to variables, one can find the optimal ordering of the variables via a simple greedy algorithm.

\begin{lemma}\label{lem:sorting}
Given an assignment of thresholds $\tau$ the optimal policy $\Pi = (\pi, \tau)$ with those thresholds is the one where $\lambda_1 \geq \lambda_2 \geq \hdots \geq \lambda_n$. 
\end{lemma}

\begin{proof}
We will show that if $\lambda_i < \lambda_{i+1}$ the swapping the order in which we visit $X_{\pi(i)}$ and $X_{\pi(i+1)}$ can only increase the reward of the policy. Let $V_{<i}$ be the reward the policy obtains from the first $i-1$ variables and the $V_{>i+1}$ be the reward it obtains from the last $i+1$ variables conditioned on reaching the $i+1$ variable. Also let $q$ be the probability of reaching the $i$-th variable. Then:
$$U(\Pi) = V_{<i} + q (1-p_i)\lambda_i + q p_i (1-p_{i+1}) \lambda_{i+1} + q p_i p_{i+1} V_{>i+1}$$
If instead we swap the order of  $X_{\pi(i)}$ and $X_{\pi(i+1)}$ we obtain:
$$U'(\Pi) = V_{<i} + q (1-p_{i+1})  \lambda_{i+1} + q p_{i+1} (1-p_i) \lambda_{i} + q p_i p_{i+1} V_{>i+1}$$
And:
$$U'(\Pi) - U(\Pi) = q (1-p_i)(1-p_{i+1}) (\lambda_{i+1}-\lambda_i) \geq 0$$
\end{proof}

One useful consequence of Lemma \ref{lem:sorting} is that if we know the \textit{thresholds} associated with each random variable (i.e. the values $\tau(i)$ for each $i$), we can efficiently compute the optimal ordering (contrast this with Lemma \ref{lem:opt_policy}, which shows how to compute the thresholds given the ordering). This will be of particular use to us in the following sections.

\subsection{Optimal orderings for $\eps$-small variables}\label{sec:small_ordering}

We are now ready to present our algorithm for $\eps$-small variables. Our main approach will be to think about as a type of ``assignment problem'', where we want to assign each random variable to a specific discrete threshold (one of the $\lceil 1/\eps \rceil$ multiples of $\eps\MAX(\bX)$). By applying Lemma \ref{lem:sorting}, we can recover the optimal ordering from the optimal assignment.

Unfortunately, computing this optimal assignment is still hard (there are on the order of $(1/\eps)^n$ possible assignments). To get around this, we construct a convex relaxation of this assignment problem which we can solve efficiently. We then show that (for $\eps$-small random variables) it is possible to round this to an integral solution without much loss in solution quality.

\paragraph{Convex programming formulation}
By Lemmas \ref{lem:discretizedthresholds} and \ref{lemma:max_threshold} it is enough for an $\eps$-optimal solution to consider discretized thresholds of the form $t_u = (1-\eps u) \MAX$ for $u \in \{0, \hdots, c\}$ for $c = \lceil 1/\eps \rceil$. Define the expected value we can obtain from a certain variable when subject to that threshold and the probabilities of not meeting the threshold for each pair:
\begin{equation}\label{eq:lambda_p}
\lambda_{i, j} = \E[X_{i} | X_{i} \geq t_{j}] \qquad  \qquad p_{i,j} = \Pr[X_{i} < t_{j}].
\end{equation}
Lemma \ref{lem:sorting} implies if we know the assignment of thresholds to variables we can obtain the best possible ordering for that assignment by sorting by $\lambda_{i, j}$. Formally, if for policy $\Pi$, variable $X_i$ occurs with threshold $t_j$ and variable $X_{i'}$ occurs with threshold $t_{j'}$, if $\lambda_{i, j} \geq \lambda_{i', j'}$ then $X_i$ should occur earlier in the ordering than $X_{i'}$. 

To encode this in the optimization, it will be convenient to define the following notation: let $\r: [n(c+1)] \rightarrow \{1, 2, \dots, n\} \times \{0, 1, \dots, c\}$ be a bijection such that
$$\lambda_{\r(1)} \geq \lambda_{\r(2)} \geq \dots \lambda_{\r(n(c+1))}.$$

We are now ready to define a convex relaxation of the assignment problem. Consider the following convex program in the variables $z_{i, j}$ (for $1 \leq i \leq n, 0 \leq j \leq c$):

\begin{equation}\label{eq:convex_relaxation}\tag{CP}
\begin{aligned}
& \max_z \sum_{\ell = 1}^{n(c+1)} \lambda_{\r(\ell)} \cdot \left(\prod_{\ell' = 1}^{\ell-1} [p_{\r(\ell')}]^{z_{\r(\ell')}} - \prod_{\ell' = 1}^{\ell} [p_{\r(\ell')}]^{z_{\r(\ell')}} \right) \quad {s.t.} \\
& \qquad \sum_{j=0}^c z_{i,j} = 1, \forall i \in [n] \\
& \qquad z_{ij}\geq 0, \forall i \in [n], j \in \{0, \hdots, c\}\\
\end{aligned}   
\end{equation}

Note that this is an assignment problem with a non-linear objective. Each variable $z_{i, j} \in [0, 1]$ represents a fractional assignment of variable $X_i$ to the threshold $t_j$. We begin by showing that the objective function for an integral assignment indeed is equal to the value of the best policy that assigns corresponding thresholds.
\begin{lemma}\label{lem:policy_to_solution}
Given any integral solution $z_{i, j} \in \{0,1\}$, let $\Pi$ be the policy that assigns thresholds $t_j$ to variable $X_i$ for each $(i,j)$ pair such that $z_{i,j}=1$ and orders the pairs according to $\lambda_{i,j}$. Then the value $U(\Pi)$ is the same as the objective function of \eqref{eq:convex_relaxation}.
\end{lemma}

\begin{proof}
We can re-write each term of the objective function as: 
$$\lambda_{\r(\ell)} \cdot \left(1 - [p_{\r(\ell)}]^{z_{\r(\ell)}}  \right) \cdot \prod_{\ell' < \ell} [p_{\r(\ell')}]^{z_{\r(\ell')}}$$
which correspond to each of the terms in the definition of $U(\Pi)$ in equation \eqref{eq:utilpi}. In particular, note that (i) when $z_{\r(\ell)} = 0$, this contribution is $0$, and (ii) the product over $\ell' < \ell$ contains exactly the variable/threshold pairs which occur earlier in the ordering $\pi$ than $\ell$ (since we are only considering orderings consistent with $\r$). 
\end{proof}

We next show that this convex program is indeed convex (in particular, we are maximizing a concave function in the $z_{i, j}$). 

\begin{lemma}\label{lem:convexity}
The objective function of \eqref{eq:convex_relaxation} is concave in the variables $z_{i,j}$.
\end{lemma}
\begin{proof}
We can rewrite the objective function as follows (assuming $\lambda_{\r(n(c+1)+1)} = 0$ for notational convenience):
\begin{align*}
\lambda_{\r(1)} - \sum_{\ell=1}^{n(c+1)}   \left( \lambda_{\r(\ell)}  - \lambda_{\r(\ell+1)}\right) \cdot \prod_{\ell' = 1}^{\ell} [p_{\r(\ell')}]^{z_{\r(\ell')}} .
\end{align*}

Since the terms $\left( \lambda_{\r(\ell)}  - \lambda_{\r(\ell+1)}\right)$ are nonnegative by the definition of $\r$ and the term $\prod_{\ell' = 1}^{\ell} [p_{\r(\ell')}]^{z_{\r(\ell')}} = \exp(\sum_{\ell' = 1}^{\ell}  z_{\r(\ell')} \log p_{\r(\ell')})$ is convex, it follows that the objective function is concave.
\end{proof}

Finally, we show that we can round a fractional solution to \eqref{eq:convex_relaxation} without losing too much reward. This is where it is essential that our variables are all $\eps$-small.

\begin{lemma}[Randomized Rounding]\label{lemma:randomized_rounding}
Let $\{z_{i, j}\}$ be a fractional solution to program \eqref{eq:convex_relaxation} and $\Obj(z)$ be its corresponding objective value. Let $\{Z_{i, j}\}$ be a random integral solution obtained by choosing for each $i$ an index $j$ with probability proportional to $z_{i, j}$, and then setting $Z_{ij}$ to $1$ and $Z_{ij'}=0$ for $j' \neq i$. Let $\Obj(Z)$ be the value of its corresponding objective.

Then if all random variables in the instance are $\eps$-small, then $\E[\Obj(Z)] \geq (1-O(\epsilon)) \cdot \Obj(z)$.
\end{lemma}

\begin{proof}
Similarly as in the proof of Lemma \ref{lem:convexity}, rewrite the objective function $\Obj(z)$ in the form
\begin{eqnarray*}
\Obj(z) &=& \lambda_{\r(1)} - \sum_{\ell=1}^{n(c+1)}   \left( \lambda_{\r(\ell)}  - \lambda_{\r(\ell+1)}\right) \cdot \prod_{\ell' = 1}^{\ell} [p_{\r(\ell')}]^{z_{\r(\ell')}} \\
&=& \sum_{\ell=1}^{n(c+1)}   \left( \lambda_{\r(\ell)}  - \lambda_{\r(\ell+1)}\right) \cdot \left(1 - \prod_{\ell' = 1}^{\ell} [p_{\r(\ell')}]^{z_{\r(\ell')}} \right).
\end{eqnarray*}

Similarly, we have that
\begin{equation} \label{eq:objZ}
\Obj(Z) = \sum_{\ell=1}^{n(c+1)}   \left( \lambda_{\r(\ell)}  - \lambda_{\r(\ell+1)}\right) \cdot \left(1 - \prod_{\ell' = 1}^{\ell} [p_{\r(\ell')}]^{Z_{\r(\ell')}} \right).
\end{equation}

Since all variables $X_i$ are $\eps$-small, we have that $p_{\r(\ell)} \in [1-\eps, 1]$ for all $\ell$. Note that for $p \in [1-\eps, 1]$, and a Bernoulli random variable $Z$ with probability $z$, it is true that $\E[p^{Z}] = (1-z) + zp \leq e^{-(1-p)z} = p^{z ((1-p)/\log 1/p)} = p^{(1-O(\eps))z}$. Secondly, note that while the variables $Z_{i,j}$ are not independent, they are negatively associated (since $\sum_{j}Z_{i, j} = 1$ must hold; see e.g. \cite{joag1983negative}). We therefore have that (for any $1 \leq \ell \leq n(c+1)$)  
\begin{equation}\label{eq:prod_ineq}
\E\left[\prod_{\ell' = 1}^{\ell} [p_{\r(\ell')}]^{Z_{\r(\ell')}}\right] \leq \left(\prod_{\ell' = 1}^{\ell} [p_{\r(\ell')}]^{z_{\r(\ell')}}\right)^{1-O(\eps)}
\end{equation}
where the expectation is over the randomized rounding. Substituting this into \eqref{eq:objZ} (and using the fact that $1 - x^{1-\eps} \geq (1-x)(1-\eps)$ for $x, \eps \in [0, 1]$), we have that
\begin{eqnarray*}
\E[\Obj(Z)] &=& \sum_{\ell=1}^{n(c+1)}   \left( \lambda_{\r(\ell)}  - \lambda_{\r(\ell+1)}\right) \cdot \left(1 - \left(\prod_{\ell' = 1}^{\ell} [p_{\r(\ell')}]^{z_{\r(\ell')}}\right)^{1-O(\eps)} \right) \\
&\geq& (1-O(\eps))\sum_{\ell=1}^{n(c+1)}   \left( \lambda_{\r(\ell)}  - \lambda_{\r(\ell+1)}\right) \cdot \left(1 - \left(\prod_{\ell' = 1}^{\ell} [p_{\r(\ell')}]^{z_{\r(\ell')}}\right) \right) \\
&=& (1-O(\eps))\Obj(z).
\end{eqnarray*}
\end{proof}

We now have all the ingredients to prove Theorem \ref{thm:small_ordering}. The algorithm follows naturally from the previous lemma: we solve program \eqref{eq:convex_relaxation}, then obtain an integral solution via randomized rounding, allocate the thresholds to variables according to the integral solution and then we order the variables/threshold pairs according to $\lambda_{i,j}$.

\begin{proof}[Proof of Theorem \ref{thm:small_ordering}]
By Lemma \ref{lem:discretizedthresholds} there is a policy using only thresholds in $t_1, \hdots, t_c$ that achieves $(1-2\eps) \OPT_\free(\bX)$. Hence by Lemma \ref{lem:policy_to_solution} the optimal value $\Obj^*$ of program \eqref{eq:convex_relaxation} should be at least $(1-2\eps) \OPT_\free(\bX)$ since it has a feasible integral solution with at least that value. By applying the rounding technique in Lemma \ref{lemma:randomized_rounding} we obtain an integral solution of value at least $(1-O(\eps)) \Obj^* = (1-O(\eps)) \OPT_\free(\bX)$. Again by Lemma \ref{lem:policy_to_solution} we can obtain a policy achieving that value by assigning thresholds according to this integral solution and sorting by $\lambda_{i,j}$.

For the running time, note that the concave program has $O(n/\epsilon)$ variables and constraints and the remaining randomized rounding step and sorting take time $O(n/\epsilon)$ and $O(n \log n)$ respectively. Hence the total running time is $O(\poly(n/\epsilon))$. 
\end{proof}

\subsection{Optimal orderings for general variables}\label{sec:general_ordering}

We will solve the general case in two steps. First we will adapt the algorithm in the previous step to handle the case where all but $\tilde{O}(1/\eps^2)$ variables are $\eps$-small. We will then use our decomposition technique to reduce the general case to this setting.

\paragraph{Step 1: Constant number of non-$\eps$-small variables}

We first describe how to modify the algorithm of Section \ref{sec:small_ordering} so that it works in the presence of up to $\tilde{O}(1/\eps^{2})$ non-$\eps$-small variables. We will accomplish this by enumerating over all assignments of thresholds to these variables and showing we can incorporate this additional information into the convex program. 

Specifically, assume we have an instance $\bX = \{X_1, X_2, \dots, X_n\}$ where variables $X_{1}$ through $X_{n-k}$ are $\eps$-small (for some $k = \tilde{O}(1/\eps^2)$). Let $f:\{n-k+1, \dots, n\} \rightarrow \{0, \dots, c\}$ be an arbitrary assignment mapping the $k$ non-$\eps$-small variables to one of the $c$ thresholds $t_{j}$. Note that there are $(c+1)^{k} \approx (1/\eps)^{1/\eps^2}$ such functions $f$. 

Now, for a fixed $f$, consider the following convex program in the variables $z_{i, j}$:

\begin{equation}\label{eq:convex_relaxation_gen}\tag{CP$_f$}
\begin{aligned}
& \max_z \sum_{\ell = 1}^{n(c+1)} \lambda_{\r(\ell)} \cdot \left(\prod_{\ell' = 1}^{\ell-1} [p_{\r(\ell')}]^{z_{\r(\ell')}} - \prod_{\ell' = 1}^{\ell} [p_{\r(\ell')}]^{z_{\r(\ell')}} \right) \quad {s.t.} \\
& \qquad \sum_{j=0}^c z_{i,j} = 1, \forall i \in [n] \\
& \qquad z_{i, f(i)} = 1, \forall n-k+1 \leq i \leq n \\
& \qquad z_{ij}\geq 0, \forall i \in [n], j \in \{0, \hdots, c\}\\
\end{aligned}   
\end{equation}

Note that this convex program is equivalent to convex program \eqref{eq:convex_relaxation}, with the exception that the variables $z_{i, j}$ for $n-k+1 \leq i \leq n$ have been fixed. This corresponds to fixing the assignment of the random variables $X_{n-k+1}, \dots, X_{n}$. 

The logic in Section \ref{sec:small_ordering} continues to hold with minimal modification. Integral solutions of \eqref{eq:convex_relaxation_gen} continue to map one-to-one to the rewards of policies assigning the respective thresholds (as in Lemma \ref{lem:policy_to_solution}) and randomized rounding still produces an integral solution with value at least $(1-O(\eps))$ of the optimal objective (as in Lemma \ref{lemma:randomized_rounding}). The only difference is in the proof of Lemma \ref{lemma:randomized_rounding}, which requires the variables $z_{i,j}$ to be small in inequality \eqref{eq:prod_ineq}. However, since the $z_{ij}$ are fixed to be integral for the variables that are not small, $Z_{ij}=z_{ij}$ for those variables, and hence \eqref{eq:prod_ineq} continues to hold.

\begin{lemma}\label{lemma:ordering_few_non_small}
Let $\bX = \{X_1, \hdots, X_n\}$ be a collection of random variables where at most $\tilde{O}(1/\eps^2)$ are not $\eps$-small, then we can compute a policy in time $O(\exp(\tilde{O}(1/\eps^2)) \poly(n/\eps))$ that achieves an expected reward of $(1-O(\eps)) \OPT_\free(\bX)$.
\end{lemma}

\begin{proof}
Let $k = \tilde{O}(1/\eps^2)$ be the number of variables that are not $\eps$-small. By Lemma \ref{lem:discretizedthresholds} there is a policy using only thresholds in $t_1, \hdots, t_c$ that achieves $(1-2\eps) \OPT_\free(\bX)$. Let $f$ be the assignment of non-$\eps$-small variables to thresholds in this policy. By Lemma \ref{lem:policy_to_solution} the optimal value $\Obj^*$ of program \eqref{eq:convex_relaxation_gen} should be at least $(1-2\eps) \OPT_\free(\bX)$ since it has a feasible integral solution with at least that value. By applying the rounding technique in Lemma \ref{lemma:randomized_rounding} we obtain an integral solution of value at least $(1-O(\eps)) \Obj^* = (1-O(\eps)) \OPT_\free(\bX)$. Again by Lemma \ref{lem:policy_to_solution} we can obtain a policy achieving that value by assigning thresholds according to this integral solution and sorting by $\lambda_{i,j}$.

To find $f$, we loop over all $(1/\eps)^{1/\eps^2} = \exp(\tilde{O}(\eps^{-2}))$ possibilities for $f$. Since solving the convex program for a given $f$ takes time $\poly(n/\eps)$, this implies the desired time complexity.
\end{proof}

\paragraph{Step 2: General case via variable decomposition}
Finally, we reduce the general case to the case where only all but a few variables are $\eps$-small. This follows from the same machinery developed in Theorems \ref{thm:nearby_weak_intro} and \ref{thm:nearby_weak}. We will need slighly different versions of those theorems which we prove below. The first lemma is a variant of the coupling argument in Lemma \ref{lem:freqcoupling} for $\OPT_{\free}$ instead of $\MAX$. See Appendix \ref{apx:optimal_ordering} for the proof.

\begin{lemma}\label{lem:optcoupling}
Let $\bX = \{X_1, X_2, \dots, X_n\}$ be a collection of $n$ random variables. Let $\bX_{\mathrm{sub}} \subset \bX$ be any subset of $k$ of these random variables, and let $\bX_{\rand} \subseteq \bX_{\mathrm{sub}}$ be a uniformly randomly chosen subset of $\bX_{\mathrm{sub}}$ of size $r$. Let $\bX' = \bX \setminus \bX_{\rand}$. Then

\begin{equation}\label{eq:opt_free_approx}
\E[\OPT_{\free}(\bX')] \geq \left(1 - \frac{r}{k}\right)\OPT_{\free}(\bX).
\end{equation}
\end{lemma}

We are now ready to show the proof of the main result. The only missing step is the reduction from the general case to the case where only a few variables are not $\eps$-small:

\begin{proof}[Proof of Theorem \ref{thm:general_ordering}]
Our proof will proceed very similarly to the application of decomposition in Theorems \ref{thm:nearby_weak_intro} and \ref{thm:nearby_weak}. Fix an instance $\bX = \{X_1, X_2, \dots, X_n\}$ of non-negative random variables. As in Theorem \ref{thm:nearby_weak}, let

$$Y_i(t) = \max(X_i, t) \text{ and } Z_i(t) = Y_i(t) - t.$$ 

As we increase $t$, more of the variables $Z_i(t)$ become $\eps$-small. Fix $k = \Theta(\eps^{-2}\log \eps^{-1})$, and let $t^*$ be the supremum over all $t$ such that at most $k$ of the variables $Z_i(t^*)$ are \textit{not} $\eps$-small. This means that for $k$ values of $i$, $\Pr[Z_i(t^*) > 0] \geq \eps$, and for the remaining indices $Z_i(t^*)$ is $\eps$-small.

Let $\bX_{\bbig}$ be the set of $X_i$ corresponding to these $k$ indices. Our policy $\Pi$ will proceed as follows. We begin by sampling a random subset of $r = \eps k = \Theta(\eps^{-1}\log(\eps^{-1}))$ variables from $\bX_{\bbig}$ and call this subset $\bX_{\rand}$. Let $\bX' = \bX \setminus \bX_{\rand}$, and let $\bZ'$ be the collection of $Z_i(t^*)$ corresponding to $\bX'$. 

By construction $\bZ'$ contains at most $\tilde{O}(\eps^{-2})$ variables that are not $\eps$-small. Using Lemma \ref{lemma:ordering_few_non_small}, construct a policy $\Pi'$ that achieves an expected reward of $(1-\eps)\OPT_{\free}(\bZ')$ on $\bZ'$. Our policy begins by processing all the elements in $\bX'$ in the order their corresponding elements in $\bZ'$ would be processed by $\Pi$. For each $X_i$ we see, we compute $Z_i = \max(X_i, t^*) - t^*$ and feed it to $\Pi'$. If $\Pi'$ accepts, we accept the corresponding $X_i$; otherwise we continue onwards (we assume without loss of generality here that $\Pi'$ never accepts an element with zero reward, so this process only accepts variables $X_i \geq t^*$). If $\Pi'$ makes it through all of $\bZ'$ without accepting, we then begin processing the elements in $\bX_{\rand}$. We accept the first element we see greater than or equal to $t^*$.

We will show that $\E[U(\Pi)] \geq (1-\eps)\OPT_{\free}(\bX)$. To do so, first note that since each $X_i \in \bX_{\rand}$ also belongs to $\bX_{\bbig}$, it satisfies $\Pr[X_i > t^*] \geq \eps$. Then

$$\Pr[\max(\bX_{\rand}) \leq t^*] \leq (1-\eps)^{|\bX_{\rand}|} = (1-\eps)^{r} = (1 - \eps)^{\Theta(\eps^{-1}\log\eps^{-1})} \leq \eps.$$

It follows that with probability at least $(1-\eps)$, our policy is guaranteed to receive reward at least $t^*$ (since we only accept $X_i$ if $X_i \geq t^*$). Conditioning on this event, we can see that the additional reward over $t^*$ generated by $\Pi$ can be lower bounded by the performance of $\Pi'$; that is, $$\E[U(\Pi)] \geq (1-\eps) t^* + \E[U(\Pi')].$$ 
Now, by construction of $\Pi$, we have that

$$\E[U(\Pi')] \geq (1-\eps)\OPT_{\free}(\bZ').$$

On the other hand, note that since each $Z_i \geq X_i - t^*$, we have that 
$$t^* + \OPT_{\free}(\bZ') \geq \OPT_{\free}(\bX').$$ 
Finally, by Lemma \ref{lem:optcoupling}, we have that
$$\OPT_{\free}(\bX') \geq \left(1 - \frac{r}{k}\right)\OPT_{\free}(\bX) = (1 - \eps)\OPT_{\free}(\bX).$$
Putting the last four equations together, we have that

$$\E[U(\Pi)] \geq (1-\eps)^2\OPT_{\free}(\bX) = (1 - O(\eps))\OPT_{\free}(\bX).$$

This shows that the policy constructed is $\eps$-optimal. Note also that this reduction can be performed in linear time, so the running time is the same as that for solving the instance in Lemma \ref{lemma:ordering_few_non_small}.
\end{proof}

\bibliographystyle{plainnat}
\bibliography{prophets}

\appendix

\section{Missing proofs from Section \ref{sect:small}}

\subsection{Proof of Lemma \ref{lemma:iid_splitting}}

\begin{proof}[Proof of Lemma \ref{lemma:iid_splitting}]
To see that $\MAX(\bX) = \MAX(\bY)$ observe that the distribution of $$\Pr[\max_i X_i \leq t] = F(t)^n = (F(t)^{1/k})^{nk} = \Pr[\max_i Y_i \leq t].$$ Since $\max_i Y_i$ and $\max_i X_i$ have the same c.d.f., they should have the same expectation.

To show that $\OPT(\bX) \geq \OPT(\bY)$ consider a stopping time $\tau$ such that $\E[Y_\tau] = \OPT(\bY)$. Now we will construct a stopping time $\tau'$ such that $\E[X_{\tau'}] = \E[Y_\tau]$.

Let $\bY_k$ be the distribution over $k$-element sequences $(Y_1, Y_2, \dots, Y_k)$ where each element $Y_i$ is i.i.d. with c.d.f. $F^{1/k}$. Let $\bY_k(x)$ be the resulting distribution of $\bY_k$ conditioned on $\max Y_i = x$. Consider the following procedure for generating a stopping time $\tau'$. For each element $X_i$, we will generate elements $Y_{(i-1)k + 1}$ through $Y_{ik}$ by sampling from $\bY_k(X_i)$. If the stopping time for $Y_i$ is $\tau = j$ where $(i-1)k+1 \leq j \leq ik$, then we will let the stopping time for $X_i$ be $\tau' = i$. Note that this is a valid stopping time for $X_i$, since the values of $Y_1$ until $Y_j$ only depend on the values of $X_1$ up to $X_i$. Moreover, note that (over all random choices of $X_i$) each $Y_j$ is distributed independently according to $F^{1/k}$. Finally, note that $X_{\tau'} = X_i = \max_{(i-1)k + 1 \leq j \leq ik} Y_j \geq Y_{\tau}$, so $\E[X_{\tau'}] \geq \E[Y_{\tau}]$. 
\end{proof}

\begin{proof}[Proof of Lemma \ref{lemma:kertz_small}]
Let $F$ be the c.d.f. of the i.i.d. variables $\bX$ in Theorem \ref{thm:kertz_upper_bound}. Assume without loss of generality that $F(0) > 0$ (if not, we can modify $X_i$ so that it equals $0$ with probability $\delta' \ll \delta$ while preserving the inequality).

If these random variables $X_i$ are $\eps$-small, then we are done. If not, note that for some sufficiently large $k$, a random variable with c.d.f. $F^{1/k}$ is $\eps$-small. In particular, it suffices to take $k \geq \log F(0)/\log (1-\eps)$ to ensure that $F(0)^{1/k} \geq (1-\eps)$. Now consider a set $\bY$ of $nk$ random variables distributed with c.d.f. $F^{1/k}$. By Lemma \ref{lemma:iid_splitting}, if we had 
$\OPT(\bY) > (\beta+\delta) \MAX(\bY)$ then we would have also $\OPT(\bX) > (\beta+\delta) \cdot \MAX(\bX)$ violating Kertz's upper bound.
\end{proof}

\begin{proof}[Proof of Lemma \ref{lemma:differential_equation}]
  Consider the following ordinary differential equation problem:
  \begin{equation}\label{eq:prob_b}
    \frac{dy}{dt} = y (\log y -1) - (\beta^{-1} -1)
    \quad \text{and} \quad y(0) = 1
  \end{equation}
  It is simpler to solve for the inverse function of $y(t)$ which we denote by 
  $t(y)$. We know that:
  $$\frac{dt}{dy} = [ y (\log y -1 ) - (\beta^{-1} -1) ]^{-1} $$
  It follows that the solution can be obtained by simple integration:
  $$t(y) = t(1) - \int_y^1 [ y (\log y -1 ) - (\beta^{-1} -1) ]^{-1} dy$$
  By the Kertz equation, if we take $\beta$ to be the Kertz bound, then $t(0) =
  1$. Therefore the inverse $y(t)$ satisfies $y(0) = 1$, $y(1) = 0$ and equation
  \eqref{eq:prob_b}. Finally, we show that it also satisfies condition
  \eqref{eq:prob_a}. To see this, take the derivative of equation \eqref{eq:prob_b},
  obtaining:
  $$y'' = y' (\log y - 1) + y \frac{y'}{y} = y' \cdot \log y$$
  which can be re-written as:
  $$\log(y) = [\log(-y')]'$$
  Notice that since $y' < 0$ we need to write $\log(-y')$ instead of $\log(y')$.
  Integrating from $0$ to $t$ each expression, we get:
  $$\int_0^t \log y(s) ds = \log(-y'(t)) - \log(-y'(0))$$
  Since $y'(0) = -\beta^{-1}$, by replacing $y(0) = 1$ in equation
  \eqref{eq:prob_b}, we have that:
  $$\int_0^t \log y(s) ds = \log(-y'(t)) - \log(-\beta^{-1}) = \log (-\beta \cdot
  y'(t))$$ 
  which is exactly the condition in equation \eqref{eq:prob_a}.
\end{proof}

\subsection{Proof of Theorem \ref{thm:upper_bound}}\label{apx:proof_thm_upper_bound}

\paragraph{Structure of the optimal policy in the limit} 
For any fixed $n$ the optimal policy consists of a sequence of thresholds above which we accept each variable. For an i.i.d. instance the thresholds depend only on what fraction of the instance was already processed.  In the limit as $n\rightarrow \infty$, we may have a continuous variable $t \in [0,1]$ that we think of as the fraction of variables already observed.  The optimal policy may then be written as a function of the continuous variable $t$. Hence the optimal policy obtained via dynamic programming and the optimal time-based policy coincide.

\paragraph{Infinitesimal variables} In the limit when $n \rightarrow \infty$ the variables become $\epsilon$-small for $\epsilon \rightarrow 0$ hence the expressions of $w$ and $\tilde{w}$ (equations \eqref{eq:w} and \eqref{eq:tw} in Section \ref{sec:small_notation}) coincide. Similarly the expressions of $R$ and $\tilde{R}$ do coincide.
For that reason the formula for $\ALG$ in Lemma \ref{lemma:policy_bound} holds with equality for $\epsilon=0$. Another by-product of the proof is that the value obtained by a policy $r(t)$ from the interval $[t,1]$ assumed it hasn't terminated by time $t$ is:
\begin{equation}\label{eq:alg_t1}
\ALG([t,1]) =\int_t^1 R(r(s)) \cdot \exp \left( \ \int_t^s \log F(r(u)) du \right) ds
\end{equation}

\paragraph{Equation for the Optimal Policy}
The optimal time-based policy can be obtained by backwards induction following the same priciple of the dynamic programming solution for the optimal policy: at time $t$ one should accept any variable with value larger than what the policy can hope to obtain by using the optimal policy on $[t,1]$. Formally, the optimal policy must satisfy:
\begin{equation}\label{eq:optimal_policy}
r(t) = \int_t^1 R(r(s)) \cdot \exp \left( \ \int_t^s \log F(r(u)) du \right) ds
\end{equation}
It is more convenient to work with it in differential form. Derivating the expression above and replacing the definition of $R$ (equation \eqref{eq:defR}) we obtain the following:
\begin{equation}\label{eq:optimal_policy_diff}
r'(t) = \int_{r(t)}^\infty \log F(u) du \quad \text{ s.t. } \quad r(1) = 0
\end{equation}

Now we show that $r^*(t)$ defined in equation \eqref{eq:optimal_r} satisfies equation \eqref{eq:optimal_policy_diff} in $[q,1]$:
$$\begin{aligned}
\int_{r^*(t)}^\infty \log F(u) du & = \int_{r^*(t)}^{r^*(q)} \log F(u) du + \int_{r^*(q)}^\infty \log F(u) du \\ 
& = -\int_{q}^{t} \log F(r^*(t)) {r^*}'(t) dt + (H- r^*(q))\log p \\
& =  -\int_{q}^{t}  \frac{\log y(t)}{y'(t)} dt + \frac{1}{y'(q)} & \text{(by the defns of } H, r^*, F \text{)} \\
& =  -\int_{q}^{t}  \frac{ y''(t)}{(y'(t))^2} dt + \frac{1}{y'(q)}  & \text{(using } y'' = y' \log y \text{)} \\
& = \frac{1}{y'(t)}=r'(t)
\end{aligned}$$
Note that the value of the optimal policy in interval $[q,1]$ is  $r^*(q)$, hence there is no reason to choose any value other than $H$ in $[0,q]$ which is precisely the value of $r^*(t)$ in that interval.

\paragraph{Comparing $\OPT$ and $\MAX$}
Since $r^*$ is the optimal policy we know that $\ALG([q,1]) = r^*(q)$. The probability that some variable in $[0,q]$ has value $H$ is $\exp \left(\int_0^q \log F(H) ds\right) = \exp(q \log p) = p^q$. Hence the value of the optimal policy is:
$$\OPT_q = (1-p^q) \cdot H + p^q \cdot r^*(q) = (1-p^q) \cdot H - p^q  \int_q^1 \frac{1}{y'(t)} dt$$
The prophet benchmark is:
$$\MAX_q  = (H - r^*(q)) (1-p) + \int_0^{r^*(q)} 1-F(u) du $$
The last integral can be evaluated as follows:
$$\int_0^{r^*(q)} 1-F(u) du = -\int_1^q [1-F(r^*(t))] {r^*}'(t) dt = \int_q^1 \frac{y(t)-1}{y'(t)} dt$$
Taking the ratio $\OPT_q / \MAX_q$ and taking the limit as $q \rightarrow 1$ (and hence $p \rightarrow 0$ and $H \log p \rightarrow 1/y'(0)$) we obtain:
$$\lim_{q \rightarrow 0} \frac{\OPT_q}{ \MAX_q} = \frac{y'(0) \cdot \int_0^1 \frac{1}{y'(t)} dt}{1+y'(0) \int_0^1 \frac{1-y(t)}{y'(t)} dt}$$
The only thing left to show is that the expression on the RHS evaluates to $\beta$. Observing that $y'(0) = -1/\beta$ (see equation \eqref{eq:prob_a}) and re-arranging the terms, what we want to prove is equivalent to:
$$\int_0^1 \frac{1}{y'(t)} \left[ \left( 1-\frac{1}{\beta}\right) - y(t) - \beta y'(t) \right] dt = 0$$
Since $y' =  y (\log y -1) - (\beta^{-1} -1)$ by equation \eqref{eq:prob_b} the expression above simplifies to showing that:
$$\int_0^1 \frac{y(t) \log y(t)}{y'(t)} dt = 1-\beta$$
This can be shown using the fact that $y'' = y' \log y$ together with integration by parts:
$$\int_0^1 \frac{y(t) \log y(t)}{y'(t)} dt = \int_0^1 \frac{y''(t) y(t)}{[y'(t)]^2} dt = -\left.\frac{y(t)}{y'(t)}\right|_0^1 + \int_0^1 y'(t) \cdot \frac{1}{y'(t)} dt = 1-\beta$$

\section{Tightness of Kertz bound for imperfect prophets}\label{app:kertz_tight}

In this section we prove Lemma \ref{lem:nearby_tight} showing that the Kertz bound in Theorem \ref{thm:nearby} cannot be improved.

\begin{proof}[Proof of Lemma \ref{lem:nearby_tight}]
By the Kertz upper bound (see \cite{kertz1986}), for any $\eps$ and sufficiently large $n$ we can construct an i.i.d. prophets instance $\bX$ where $\OPT(\bX) < (\beta + \eps)\MAX(\bX)$. Since this is an i.i.d. prophets instance, we additionally have that $\MAX(\bX') \geq (1 - \frac{r}{n})\MAX(\bX)$ (one way to see this is that since all the random variables are identical, with probability $(n-r)/n$, the maximum value in $\bX$ will be among the variables in $\bX'$). Let $\alpha - \beta = \delta$. If we pick $\eps < \delta / 2$ and $n > r(1 + \delta / 2)$ it follows that $\OPT(\bX) < \alpha \MAX(\bX')$, as desired.
\end{proof}

\section{Missing proofs from Section \ref{sect:optimal_order}}\label{apx:optimal_ordering}

\begin{proof}[Proof of Lemma \ref{lem:opt_policy}]
To show that the optimal policy is stateless, we proceed inductively. The optimal policy when there is only one random variable is stateless (take it). Assume the optimal policy for any collection of $n-1$ random variables is stateless. Now, the optimal policy for a collection of $n$ random variables must start by choosing one of the variables, looking at it, and deciding whether to take it. If we do not take it, then (since the variables are all independent) the optimal continuation is to run the optimal policy for the remaining $n-1$ variables (which, by induction, is stateless). 

To decide whether to take the first value, assume we are examining variable $X_i$, and let $V_i = \OPT_{free}(\bX_{-i})$ be the value of the optimal policy on the remaining $n-1$ other variables. If $X_i \geq V_i$, then we maximize our utility by taking $X_i$; otherwise, we should pass on $X_i$ and receive $V_i$ utility in expectation. Note that this is a threshold policy. By picking the starting variable $X_i$ that leads to the highest expected utility, it follows that the optimal policy is stateless, completing the induction.

Indeed, this argument shows something stronger: at any point, you should accept an item iff its value is larger than the expected value you would receive on the remaining items. This is captured in equation \eqref{eq:opt_tau} -- the RHS is the expected reward of $\Pi$ when started at the $(k+1)$th element (compare with \eqref{eq:utilpi}). 
\end{proof}

\begin{proof}[Proof of Lemma \ref{lem:discretizedthresholds}]
Let $\Pi$ be the optimal policy for a collection $\bX$ of random variables, and let $c_1$ be the smallest integer such that
\[
(1-c_1\eps)\MAX(\bX) \leq \tau(\pi(1)).
\]
Let $\Pi_1$ be the policy that processes the variables in the same order as $\Pi$ except sets the first threshold to $\tau'(\pi(1)) = (1-c_1\eps)\MAX(\bX)$.  Let us compare the expected rewards $U(\Pi)$ and $U(\Pi_1)$. Note that as a consequence of Lemma \ref{lem:opt_policy} and equation \eqref{eq:utilpi}, we have that

\begin{equation}\label{eq:U}
    U(\Pi) = \lambda_1(1-p_1) + \tau(\pi(1))p_1.
\end{equation}

Similarly, we have that (defining $\lambda'_1$ and $p'_1$ analogously to $\lambda_1$ and $p_1$)

\begin{equation}\label{eq:Up}
    U(\Pi_1) = \lambda'_1(1-p'_1) + \tau(\pi(1))p'_1.
\end{equation}

(Note that the term $\tau(\pi(1))$ in \eqref{eq:Up} is correct and should \textit{not} be $\tau'(\pi(1))$; this represents the expected reward from the remainder of the protocol). Now, note that 

$$\lambda'_1(1-p'_1) - \lambda_1(1-p_1) = \E[X_{\pi(1)} \cdot \mathbbm{1}\{\tau'(\pi(1)) \leq X_{\pi(1)} \leq \tau(\pi(1))\}].$$
and that
$$p_1 - p'_1 = \Pr[\tau'(\pi(1)) \leq X_{\pi(1)} \leq \tau(\pi(1))].$$
It follows that

\begin{eqnarray*}
U(\Pi) - U(\Pi_1) &=& \E[(\tau(\pi(1)) - X_{\pi(1)}) \cdot \mathbbm{1}\{\tau'(\pi(1)) \leq X_{\pi(1)} \leq \tau(\pi(1))\}].\\
&\leq & (\tau(\pi(1)) - \tau'(\pi(1)))\Pr[\tau'(\pi(1)) \leq X_{\pi(1)} \leq \tau(\pi(1))] \\
&\leq & (1-p_1)\eps\MAX(\bX).
\end{eqnarray*}

Similarly, we can consider the policy $\Pi_2$ where $\tau(\pi(2))$ is also rounded down to $(1 - c_2\eps)\MAX$ for the nearest integer $c_2$. Using a similar argument to the above, we find that 
\[
U(\Pi_1) - U(\Pi_2) \leq  p_1(1-p_2) \eps \MAX(\bX).
\]
Indeed, generalizing this to $\Pi_i$ for any $1 < i \leq n$, we can show that %
\[
U(\Pi_i) - U(\Pi_{i+1}) \leq  \left(\prod_{j<i}p_i\right)(1-p_{i})\eps \MAX(\bX).
\]

Summing all these inequalities up, we find that $U(\Pi) - U(\Pi_{n}) \leq \eps\MAX(\bX)$. Since $\MAX(\bX) \leq 2\OPT_{\free}(\bX)$ (by the classic worst-order prophet inequality of \cite{KS77}) and $U(\Pi) = \OPT_{\free}$ (since $\Pi$ is optimal), it follows that $U(\Pi_{n}) \geq (1-2\eps)\OPT_{\free}$.
\end{proof}

\begin{proof}[Proof of Lemma \ref{lem:optcoupling}]
Let $\Pi$ be an optimal policy for $\bX$. We will construct a policy $\Pi'$ for $\bX'$ whose performance satisfies \eqref{eq:opt_free_approx} in expectation.

Consider the following policy $\Pi'$. For each random variable in $\bX \setminus \bX'$, sample a virtual value for this variable. Now, run policy $\Pi$ on all $n$ random variables ($\bX'$ and the $r$ virtually sampled ones). If the policy $\Pi$ terminates on a variable in $\bX'$, return it; if it terminates on a variable in $\bX \setminus \bX'$, accept the next variable in $\bX'$ (alternatively, perform arbitrary actions for the remainder of the policy). 

We claim that $\E[U(\Pi')] \geq (1 - r/k)U(\Pi)$, thus proving the desired result. To see this, note that the only events where the execution of policy $\Pi'$ differs from that of policy $\Pi$ are the cases when $\Pi$ chooses an item we removed from $\bX$. Since we are choosing the $r$ items to remove uniformly at random out of a set of size $k$, this happens with probability at most $r/k$. Therefore, with the remaining $(1 - r/k)$ probability we receive utility $U(\Pi)$ in expectation.
\end{proof}

\end{document}